\documentclass[aps,pra,superscriptaddress,floatfix,twocolumn]{revtex4-2}

\usepackage{titlesec}
\titlespacing*{\section}{0pt}{1.2ex plus .2ex minus .1ex}{0.6ex plus .1ex}

\usepackage{float} 
\usepackage{graphicx}
\usepackage{amsthm}
\usepackage{thmtools}
\usepackage{mathtools}
\usepackage{thm-restate}
\usepackage{amsmath}
\usepackage{algorithm}
\usepackage{algorithmic}
\usepackage{amssymb}
\usepackage{bm}
\usepackage{xcolor}
\usepackage{placeins}
\usepackage[normalem]{ulem}
\usepackage[colorlinks]{hyperref}
\hypersetup{
	pdfstartview={FitH},
	pdfnewwindow=true,
	colorlinks=true,
	linkcolor=blue,
	citecolor=blue,
	filecolor=blue,
	urlcolor=blue}
\usepackage{tikz}
\usetikzlibrary{calc,backgrounds}
\usepackage{physics}
\usepackage{cancel}
\usepackage{multirow}
\usepackage{pgffor}
\usepackage{url}
\usepackage{hypcap}
\usepackage{dsfont}
\usepackage{soul}
\usepackage{inconsolata}

\usepackage[caption=false]{subfig}

\hypersetup{pdfpagemode=UseNone}

\allowdisplaybreaks

\newcommand{\thm}[1]{\hyperref[thm:#1]{Theorem~\ref*{thm:#1}}}
\newcommand{\defn}[1]{\hyperref[defn:#1]{Definition~\ref*{defn:#1}}}
\newcommand{\lem}[1]{\hyperref[lem:#1]{Lemma~\ref*{lem:#1}}}
\newcommand{\prop}[1]{\hyperref[prop:#1]{Proposition~\ref*{prop:#1}}}
\newcommand{\fig}[1]{\hyperref[fig:#1]{Figure~\ref*{fig:#1}}}
\newcommand{\tab}[1]{\hyperref[tab:#1]{Table~\ref*{tab:#1}}}
\renewcommand{\sec}[1]{\hyperref[sec:#1]{Section~\ref*{sec:#1}}}
\newcommand{\app}[1]{\hyperref[app:#1]{Appendix~\ref*{app:#1}}}
\newcommand{\cor}[1]{\hyperref[cor:#1]{Corollary~\ref*{cor:#1}}}
\newcommand{\obs}[1]{\hyperref[obs:#1]{Observation~\ref*{obs:#1}}}
\newcommand{\nn}{\nonumber \\}
\newcommand{\append}[1]{\hyperref[append:#1]{Appendix~\ref*{append:#1}}}

\usepackage{qcircuit}

\usepackage[capitalise]{cleveref} 

\newcommand{\subfigimg}[3][,]{%
  \setbox1=\hbox{\includegraphics[#1]{#3}}
  \raisebox{\dimexpr\ht1-1\baselineskip}{#2}\hspace*{10pt}
  \usebox1
}

\newtheorem{theorem}{Theorem}

\newtheorem{lemma}[theorem]{Lemma}

\usepackage{qcircuit}

\newcommand{\MQ}{\affiliation{
School of Mathematical and Physical Sciences,
Macquarie University, Sydney, NSW 2109, AU} }

\newcommand{\PKU}{\affiliation{
Beijing International Center for Mathematical Research, Peking University, Beijing, China}}

\newcommand{\CQ}{\affiliation{%
ContinoQuantum, Sydney, NSW 2093, AU}}

\newcommand{\AMZ}{\affiliation{%
AWS Center for Quantum Computing, Pasadena, CA}}

\begin{document}

\title{Constant Factor Analysis of Optimal Quantum Linear Solvers in Practice}

\author{Pedro C. S. Costa}\email[Corresponding author: ]{pcs.costa@protonmail.com}\CQ\MQ
\author{Alexander M. Dalzell}\AMZ
\author{Dong An}\PKU
\author{Dominic W. Berry} \MQ

\begin{abstract}
Optimal quantum linear equation solvers provide complexity $O(\kappa\log(1/\epsilon))$, where $\kappa$ is the condition number and $\epsilon$ is the allowable error. The optimal solver using a discrete adiabatic approach [PRX Quantum \textbf{3}, 040303 (2022)] has large analytically proven constant factors for the upper bound on the complexity. The constant factors were later found to be about 1,200 times smaller in numerical testing [Quantum \textbf{9}, 1887 (2025)]. This meant it is about an order of magnitude more efficient than using a randomised approach from [PRX Quantum \textbf{6}, 040373 (2025)], which has far smaller analytically proven constant factors. Recently, a ``Shortcut'' method has been found to provide an optimal solver which also has small proven constant factors. In the present work, we conduct a comprehensive numerical analysis comparing this method with the adiabatic solver for two families of random linear systems. We find that, in the case where the solution norm is \emph{unknown}, the adiabatic solver provides slightly better performance. If the solution norm is \emph{known}, then the shortcut method provides significantly better performance for non-Hermitian matrices.
\end{abstract}
\maketitle

\section{Introduction}

Just as linear system algorithms are foundational in classical computing—powering numerical methods for differential equations and serving as core subroutines in machine learning—quantum linear system algorithms  (QLSAs) \cite{morales2411quantum} are poised to play an analogous role in quantum computing, as primitives for quantum algorithms for differential equations \cite{Berry2024quantumalgorithm,costa2025further,liu2023efficient} and quantum machine  learning \cite{zhao2026exponential}, with polylogarithmic dependence on the system dimension (i.e., exponential improvement over classical scaling) under standard assumptions such as sparsity, moderate condition number, $\kappa$, and efficient oracle access.

The first QLSA was proposed by Harrow, Hassidim, and Lloyd \cite{Harrow_2009}, which has complexity $\mathcal{O}(\kappa^2/\epsilon)$ for allowable error $\epsilon$.
That complexity is suboptimal, because it scales as the square of the condition number, whereas the lower bound found in Ref.~\cite{Harrow_2009} was $\Omega(\kappa)$.
Variable-time amplitude amplification was proposed as a way to obtain near-linear scaling in $\kappa$, albeit at the cost of worse scaling in $\epsilon$ \cite{ambainis,ambainis2010variable}.
The complexity was further improved to being polynomial in $\log(1/\epsilon)$ using a linear combinations of unitaries approach \cite{CKS}.

A new principle in QLSAs was the development of an adiabatic solver, with the first proposal 
giving complexity $\mathcal{O}(\kappa\log(\kappa)/\epsilon)$ \cite{PhysRevLett.122.060504}.
The complexity was further improved to polynomial in $\log(1/\epsilon)$ using improved scheduling \cite{an2019quantum},
then to $\mathcal{O}(\kappa\log(\kappa/\epsilon))$ using filtering \cite{lin2020optimal}.
The optimal complexity is $\mathcal{O}(\kappa\log(1/\epsilon))$, and was achieved using a quantum walk with a discrete adiabatic theorem \cite{costa2022optimal}.
An alternative approach to achieving optimal complexity is an adiabatic solver using randomised times, though
initial methods were slightly suboptimal \cite{PhysRevLett.122.060504,jennings2023arxiv}.
This was improved to being strictly optimal scaling using ``Poissonization'' \cite{Cunningham2024}, which was incorporated into Ref.~\cite{jennings2023efficient} (the published version of Ref.~\cite{jennings2023arxiv}).
A new approach to achieving the optimal complexity is the ``Shortcut'' method proposed in Ref.~\cite{dalzell2024shortcut}.

After the development of QLSAs with the optimal $\mathcal{O}(\kappa \log(1/\epsilon))$ complexity, the focus shifted to the constant factors in this complexity. The QLSA constant factors are important when evaluating the feasibility of proposed end-to-end applications via detailed resource estimates, as has been done for portfolio optimization \cite{dalzell2022socp} and solving certain differential equations \cite{jennings2023costDiffEQ,penuel2025detailedassessmentcalculatingdrag,jennings2025endtoendNonlinearDiffEQ}. 
The discrete adiabatic quantum walk (QW) approach \cite{costa2022optimal} proved an upper bound on the complexity with a large constant factor, leaving open the possibility for significant improvements.
The ``randomised'' method for the quantum linear system problem (QLSP) was initially suboptimal, but much tighter bounds on the complexity were proven, leading to the conjecture that it would outperform the QW method for realistic parameters \cite{jennings2023efficient}.
To test this conjecture, Ref.~\cite{costa2023discrete} numerically tested both quantum solvers with thousands of random instances of matrices with a range of dimensions and condition numbers.
The results from Ref.~\cite{costa2023discrete} show that in practice, the constant factor of the QW method is about $1200$ times smaller than the upper bound estimated in Ref.~\cite{costa2023discrete}.
As a result, the QW outperforms the randomised method by roughly an order of magnitude.

More recently, the Shortcut method \cite{dalzell2024shortcut} provides an optimal quantum linear system solver.
This method is based on the quantum singular value transformation (QSVT) \cite{GilyenSuLowEtAl2019}, and offers a significantly smaller constant factor guarantee compared to the upper bound given for the QW \cite{costa2022optimal}.
In this work, we perform comparative numerical testing of the QW and Shortcut methods.
In comparison to the testing in Ref.~\cite{costa2023discrete}, we test larger condition number and dimension, and different sparsity conditions. We observe a slight different constant factor for the new range of condition number and dimension in the algorithms' performance.      We therefore test the randomised method with higher dimension, to ensure that the conclusions of Ref.~\cite{costa2023discrete} still hold.

The background of the Shortcut method is presented in \cref{sec:BackG}. 
Details of the numerical tests for the QLS solvers are given in \cref{Sec:result}. 
We conclude in \cref{sec:conc}. 
Additional technical details on the Shortcut method and further numerical analyses are provided in the appendices.

\section{Review and Algorithmic Method for the Shortcut to an Optimal QLSP Solver}
\label{sec:BackG}

\subsection{The QLSP} 

We begin with a brief review of the Quantum Linear System Problem (QLSP), based primarily on Ref.~\cite{dalzell2024shortcut}. In this setting, the input is a square matrix \( A \in \mathbb{C}^{m \times m} \), and the goal is to solve the linear system \( A\mathbf{x} = \mathbf{b} \), where \( \mathbf{x}, \mathbf{b} \in \mathbb{C}^m \).

In the quantum version, we assume access to a unitary \( U_{\mathbf{b}} \) acting on \( s \approx \log_2(m) \) qubits such that
\begin{equation}
U_{\mathbf{b}}\ket{e_0} = \ket{\mathbf{b}} = \|\mathbf{b}\|^{-1} \sum_{j=0}^{2^s -1} b_j \ket{e_j},
\end{equation}
i.e., the normalized vector \( \mathbf{b} \) is encoded as the amplitude of a quantum state in the computational basis \( \{\ket{e_j}\}_{j=0}^{2^s-1} \).

We also assume \( \|A\| = 1 \). Access to \( A \) is provided via a block encoding: an \((\alpha, a)\)-block-encoding unitary \( U_A \) acting on \( a + s \) qubits, such that
\begin{equation}
A = \alpha \left( \bra{0}^{\otimes a} \otimes I_{2^s} \right) U_A \left( \ket{0}^{\otimes a} \otimes I_{2^s} \right),
\end{equation}
with \( \alpha \geq \|A\| \).

Given an error parameter \( \epsilon \), the goal of the QLSP is to produce a quantum state \( \ket{\tilde{\mathbf{x}}} \) such that \( \|\ket{\mathbf{x}} - \ket{\tilde{\mathbf{x}}}\| \leq \epsilon \), where \( \ket{\mathbf{x}} \propto A^{-1}\ket{\mathbf{b}} \).
For non-deterministic procedures this error condition is generalised to requiring an output quantum state $\rho_\mathbf{x}$ such that the Bures distance is no greater than $\epsilon$.

Since the algorithm proposed in \cite{dalzell2024shortcut} requires an estimate of \( \|\mathbf{x}\| \), it is useful to know its maximum possible value. Given \( \kappa = \|A\| \cdot \|A^{-1}\| \) and \( \|\mathbf{b}\| = 1 \), we have
\begin{equation}
\|\mathbf{x}\| = \|A^{-1}\mathbf{b}\| \leq \|A^{-1}\| \cdot \|\mathbf{b}\| = \kappa.
\end{equation}

\subsection{The Algorithm}

The algorithm from \cite{dalzell2024shortcut} consists of two key components: a Kernel Reflection (KR) followed by a Kernel Projection (KP), detailed in \cref{sec:kernels}. In this work, we focus solely on the KR part, whose cost we analyze to extract a constant factor for comparison with the results in \cite{costa2023discrete}.

To apply the KR method, the matrix \( A \) must first be enlarged using a guess for the norm of the solution. We restrict ourselves to square matrices \( A \in \mathbb{C}^{n \times n} \), which are embedded into larger square matrices \( A_t \in \mathbb{C}^{m \times m} \), where \( m = n+1 \), based on a guess \( t \in [1,\kappa] \) for \( \|\mathbf{x}\| \). Let \( \{v_j\}_{j=0}^{n-1} \) be the canonical basis of \( \mathbb{C}^n \), where $A$ is represented as $A = \sum_{i,j=0}^{n-1} A_{ij} v_i v_j^\dagger$. Then, we define
\begin{align}
A_t &\coloneqq \sum_{i,j=0}^{n-1} A_{ij} e_i e_j^\dagger + \frac{1}{t} e_{m-1} e_{m-1}^\dagger\nonumber\\
&=
\begin{pmatrix}
A & 0 \\
0 & \tfrac{1}{t}
\end{pmatrix},
\end{align}
where \( \{e_j\}_{j=0}^{m-1} \), with $m=n+1$, is the canonical basis of \( \mathbb{C}^m \).

The vector \( \mathbf{b} = \sum_{j=0}^{n-1} b_j v_j \) is similarly extended to
\begin{equation}
\mathbf{b}' \coloneqq \sum_{j=0}^{n-1} \frac{b_j}{\sqrt{2}} e_j + \frac{1}{\sqrt{2}} e_{m-1}.
\end{equation}
The corresponding solution in the enlarged space satisfies \( A_t \mathbf{x}_t = \mathbf{b}' \), where
\begin{equation}
\label{eq:xt}
\mathbf{x}_t = \sum_{j=0}^{n-1} \frac{x_j}{\sqrt{2}} e_j + \frac{t}{\sqrt{2}} e_{m-1}.
\end{equation}

All block encodings and QSVT constructions are performed with these extended quantities. The block-encoding \( U_{A_t} \) and state-preparation unitary \( U_{\mathbf{b}'} \) can be built using \( U_A \) and \( U_{\mathbf{b}} \), as prescribed in \cite{dalzell2024shortcut}. In practice, the circuit implementing the block-encoding for $A_t$ has only one $t$-dependent gate---a single-qubit rotation where the angle of rotation is determined by $t$.

From these, a block encoding of
\begin{equation}
\label{eq:G_t}
G_t = Q_{\mathbf{b}'} A_t
\end{equation}
is constructed, where $Q_{\mathbf{b}'}= I - \mathbf{b}'\mathbf{b}'^{\dagger}$ is the projector onto the complement of $\mathbf{b}'$, which ensures that \( \mathbf{x}_t \in \ker(G_t) \); see \cref{sec:kernels}.  In fact, the kernel of $G_t$ is the span of $\mathbf{x}_t$ and the kernel of $A_t$. The QSVT is then applied to \( G_t \) to implement the kernel reflection. 

The idea of the KR method is to start from the state $\ket{e_{m-1}}$ in the $m$-dimensional Hilbert space and then rotate to the orthogonal vector $\ket{\mathbf{x}}$ where the solution to the augmented linear system satisfies
\begin{equation}\label{eq:x_t_in_terms_of_x_and_e_m-1}
    \ket{\mathbf{x}_t} = \sin(\theta_t)\ket{\mathbf{x}} + \cos{(\theta_t)}\ket{e_{m-1}},
\end{equation}
with
\begin{equation}
\theta_t = \arctan{\left(\frac{\|\mathbf{x}\|}{t}\right)}.
\end{equation}
We summarise the quantum algorithm steps from its quantum circuit in \cref{fig:Algo}. We analyse it step by step.
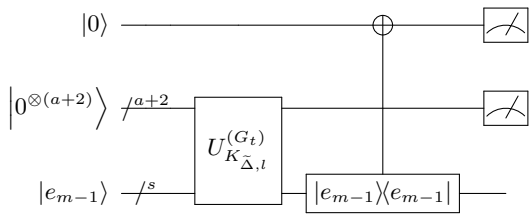
\begin{figure}[H]
\centerline{
\Qcircuit @R=2em @C=1em {
&&&&&\lstick{\ket{0}} & \qw  &\qw &\qw&\targ{}&\meter\\
&&&&&\lstick{\ket{0^{\otimes (a+2)}}}  & {/}^{a+2} \qw& \qw &\multigate{1}{U^{(G_t)}_{K_{\widetilde{\Delta},l}}}&\qw &\meter\\
&&&&&\lstick{\ket{e_{m-1}}} & {/}^s \qw& \qw &\ghost{U^{(G_t)}_{K_{\widetilde{\Delta},l}}}& \gate{\ketbra{e_{m-1}}}\qwx[-2]&\qw  
}}
\caption{\label{fig:Algo} Quantum circuit of the KR algorithm (reproduction of Figure 10 of \cite{dalzell2024shortcut}). We have $U^{(G_t)}_{K_{\widetilde{\Delta},l}}$ the QSVT that implements the kernel reflection polynomial $K_{\widetilde{\Delta},l}$ on the singular values of $G_t$. We have $a$ ancillas to build the block encoding of $U_{A_t}$, to which two more ancillas are added: one resulting from the QSVT and another from the extra ancilla used for $U_{G_t}$. }
\end{figure}

First, we initialize the algorithm by preparing the state $\ket{e_{m-1}}$ (final register of \cref{fig:Algo}), which admits the decomposition
\begin{equation}
\ket{e_{m-1}} = \cos(\theta_t)\ket{\mathbf{x}_t} + \sin(\theta_t)\ket{\mathbf{y}_t}.
\end{equation}
where
\begin{equation}
\ket{\mathbf{y}_t} = -\cos(\theta_t)\ket{\mathbf{x}} + \sin(\theta_t)\ket{e_{m-1}}.
\end{equation}
is orthogonal to $\ket{\mathbf{x}_t}$.

Next, we implement KR. This step is approximate, and its error is governed by a tunable input parameter $\eta_{\rm KR} \in (0,1]$. In the quantum circuit of \cref{fig:Algo}, it is depicted by the gate $U^{(G_t)}_{\widetilde{\Delta},l}$---a QSVT sequence consisting of several calls to the block-encoding of $G_t$ defined in \cref{eq:G_t}. The effect is the transformation by the order-$2l$ polynomial $K_{\widetilde{\Delta}, l}$ on the singular values of $G_t$, applied to the initial state $\ket{e_m}$. The polynomial and the parameters $\widetilde{\Delta}$ and $l$ are determined based on $\eta_{\rm KR}$ and $\kappa$ in order to approximate a reflection about the kernel. 
(See \cref{eq:K_delt} for the explicit form of $K_{\widetilde{\Delta}, l}$.) 
We may note that $\ket{\mathbf{x}}$ and $\ket{e_{m-1}}$ by construction lie in the orthogonal complement of the kernel of $A_t$ and meanwhile it holds that $\ket{\mathbf{x}_t}\in \ker(G_t)$  and $\ket{\mathbf{x}_t} \perp$ $\ket{\mathbf{y}_t}$. These facts imply that $\ket{\mathbf{y}_t}$ lies in the orthogonal complement of the kernel of $G_t$. Thus, application of $U^{(G_t)}_{K_{\widetilde{\Delta},l}}$, implements an approximate reflection that flips the sign of $\ket{\mathbf{y}_t}$, up to a small correction, yielding
\begin{equation}
\cos(\theta_t)\ket{\mathbf{x}_t} - (1 - \delta_1')\sin(\theta_t)\ket{\mathbf{y}_t} + \delta_2'\sin(\theta_t)\ket{\mathbf{z}},
\end{equation}
where $\ket{\mathbf{z}}$ is a unit
vector orthogonal to  $\ket{\mathbf{x}_t}$ and $\ket{\mathbf{y}_t}$.
The small corrections $\delta'_1$ and $\delta'_2$ are each upper bounded by $\mathcal{O}(\eta_{\rm KR})$  (see \cref{eq:etas}), and vanish as the tunable parameter $\eta_{\rm KR} \rightarrow 0$. 
Note that these expressions refer to the final register components conditioned on ancilla outcomes in the second register of \cref{fig:Algo} being zero.

Finally, the algorithm implements a multicontrolled Toffoli gate (last gate of  \cref{fig:Algo}) 
and postselects on the target ancilla being $\ket{0}$. This has the effect of projecting the final register onto the  subspace orthogonal to $\ket{e_m}$, resulting in a state proportional to:
\begin{equation}
\label{eq:fin_state}
\ket{\tilde{\mathbf{x}}} \propto \cos(\theta_t)\sin(\theta_t)(2 - \delta_1')\ket{\mathbf{x}} + \delta_2'\sin(\theta_t)\ket{\mathbf{z}}.
\end{equation}

The total probability of success, corresponding to measuring all ancillas in the $\ket{0}$ state, is given by the squared norm of the unnormalized state above:
\begin{equation}
p_{\mathrm{succ}} = \frac{1}{4}\sin(2\theta_t)^2(2 - \delta_1')^2 + (\delta_2')^2\sin^2(\theta_t).
\end{equation}
Thus, we upper bound this success probability by considering the maximum of $\delta_2'$ and the minimum of $\delta_1'$, and vice-versa for the lower bound, i.e., 
\begin{equation}
\label{eq:P_bound}
\sin^2(2\theta_t)\left(\frac{1 - \eta_{\rm KR}}{1 + \eta_{\rm KR}}\right)^2 \leq p_{\mathrm{succ}} \leq \sin^2(2\theta_t) + \frac{4\eta_{\rm KR}^2}{(1 + \eta_{\rm KR})^2},
\end{equation}
Notice that from \cref{eq:P_bound} as $\eta_{\rm KR} \rightarrow 0$ and $t\approx \|x\|$  we have $p_{\mathrm{succ}} \rightarrow 1$.

Conditioned on success, the error $\lVert \ket{\mathbf{x}} - \ket{\tilde{\mathbf{x}}}\rVert $  can also be bounded by $\arcsin(\eta_{\rm KR}/\cos(\theta_t)) \approx \eta_{\rm KR}/\cos(\theta_t)$ (see Theorem 1 of \cite{dalzell2024shortcut}). Thus, if $\epsilon$ error is desired, one can choose $\eta_{\rm KR} = \epsilon \cos(\theta_t)$ (note, knowledge of $\| \mathbf{x} \|$ is required to compute $\theta_t)$. Alternatively, the output can be considered an ansatz state, and the KP method can be used to filter errors down to a smaller value of $\epsilon$, if desired. 

\section{Methods and results}
\label{Sec:result}

When comparing the performance of quantum linear solvers, it is important to distinguish between scenarios where the norm \( \|\mathbf{x}\| \) is known and where it is not. As discussed in~\cite{dalzell2024shortcut}, the kernel reflection algorithm requires an initial estimation of the norm up to a multiplicative error. This step is not present in other methods such as those based on the discrete adiabatic theorem or the randomization approach \cite{costa2023discrete}.

In the idealised setting where the norm of the solution is known, the Shortcut method does not require the filtering step; that is, the KR procedure alone achieves the optimal scaling 
$\mathcal{O}\bigl(\kappa \log(1/\epsilon)\bigr)$
in terms of the condition number $\kappa$ and the allowed solution error $\epsilon$. In fact, when the norm is known exactly, one may choose $\eta_{\rm KR} = \epsilon/\sqrt{2}$, and then derive a rigorous query complexity upper bound of $(1+\mathcal{O}(\epsilon))\kappa \log(2\sqrt{2}/\epsilon)$,  implying the constant factor is only 1, even in the worst case. This has the same cost as KP (filtering) to error $\eta_{\rm KP} = \epsilon$ up to a small additive $\kappa \ln(\sqrt{2})$ correction. For this reason, the Shortcut method will provide the lowest complexity in the case where the norm is known. In the known-norm scenario, we do not include a filtering step, and we benchmark both QLSAs for total allowable
errors $\epsilon = 10^{-2}$ and $\epsilon = 10^{-3}$. 
We analyse this cost in two distinct cases: (1) when the input matrix $A$ is
non-Hermitian, and (2) when $A$ is positive definite (PD). In both regimes, we
test matrices of sizes $32\times 32$ and $64\times 64$.


For the QW-based method, we estimate the total cost separately for the non-Hermitian (and non-positive) and positive-definite data sets, but using essentially the same set of adiabatic threshold errors, namely $\Delta \in \{0.40,0.30,0.20,0.15,0.10\}$. Here, $\Delta$ denotes the $\ell_2$-norm error in the output state of the adiabatic stage, prior to the final filtering step that reduces the error to a prescribed target solution error $\epsilon < \Delta$. For each data set, we use the resulting runs to interpolate the total cost as a function of $\Delta$, taking into account both the adiabatic evolution and the filtering stage, and thereby estimate the value of $\Delta$ that minimizes the total cost, (see Eq. (27) \cite{costa2023discrete}) for a given allowable final error $\epsilon$.

The depth of the circuit in the Shortcut method—and consequently the total cost—is determined by selecting the parameter $\eta_{\rm KR}$ in the polynomial defining the kernel reflection such that the resulting allowable error in the output state is as close as possible to the target error. For the known-norm case, we set the target error to $\epsilon$ as there is no subsequent filtering step. (Later, we will consider a two step approach, where the target error is $\Delta$, after which filtering is applied to reduce error to $\epsilon$.) Once $\eta_{KR}$ is fixed, and for a given value of $\widetilde{\Delta}:= \frac{1}{\kappa}$ in the kernel reflection function, the corresponding polynomial order \cref{eq:order:pol} determines the circuit depth. This order plays an analogous role to the number of steps necessary in the discrete adiabatic method described in \cite{costa2023discrete}, which is required to achieve the desired target error $\Delta$. In that case, the evolution progresses from $s = 0$ to $s = 1$ according to a given scheduling function.

Subsequently, we consider the more realistic scenario in which the norm of the solution is not known in advance. In this setting, we restrict our analysis to non-Hermitian matrices, and for the Shortcut method, we use the \textit{exhaustive search method in log space} variant, as proposed in~\cite[Algorithm 3]{dalzell2024shortcut}. This variant has near-optimal asymptotic cost of $\mathcal{O}(\kappa \log(\kappa)\log\log(\kappa) + \kappa \log(1/\epsilon))$, and it is simpler to implement than the variant from \cite[Algorithm 4]{dalzell2024shortcut} that achieves optimal $O(\kappa \log(1/\epsilon))$ complexity even when the norm is not known.

For the Shortcut method, we quantify the complexity in terms of calls to the (controlled) block-encoding $U_A$ of $A$. Specifically, if $\#U_A$ denotes the polynomial degree, then we define the effective cost as $\mathrm{Cost} = \#U_A/p_{\mathrm{succ}}$, where $p_{\mathrm{succ}}$ is the probability of measuring $\ket{0}$ in the ancillary registers. In the case of the Shortcut method, the number of calls to $U_A$ is given by the polynomial order of the KR, whereas in the QW method, it is given by the number of QW steps. This assumes each QW step has cost equal to the cost of one call to (controlled) $U_A$.  However, for non-Hermitian matrices, a QW step during the adiabatic part of the QW method requires applying a gate of the form $\ket{0}\bra{0} \otimes U_A + \ket{1}\bra{1}\otimes U_{A^\dagger}$ which controls between $A$ and its adjoint---in some cases, it may be appropriate to cost this operation as double the cost of controlled $U_A$. Thus, we also report in \cref{app:QWfilt} the total cost of the QW method by doubling the cost of the adiabatic part. Results are shown for both positive definite and non-Hermitian matrices as a function of the condition number. 

For each condition number $\kappa$, we generate 100 independent random instances $(A,\mathbf{b})$ and report the average number of $U_A$ calls. We first sample a matrix $A\in\mathbb{R}^{2^n\times 2^n}$ with i.i.d.\ entries drawn uniformly from $[0,2]$, and compute its singular value decomposition as $A = U \Sigma V^\top$. To enforce the target condition number $\kappa$, we replace the singular values of $\Sigma$ by affinely rescaled singular values collected in a diagonal matrix $\Sigma'$, chosen so that $\sigma_{\max}(\Sigma')/\sigma_{\min}(\Sigma') = \kappa$. We then set $M \coloneqq U \Sigma' V^\top$ for the non-Hermitian instances. For the PD instances, we instead set $M\coloneqq V \Sigma' V^\top$, which yields a symmetric positive-definite matrix.

\subsection{Performance in the Known-Norm Regime}

It is important to note that this analysis excludes the cost of the norm estimation procedure; it accounts only for the kernel reflection steps.

\subsubsection{Non-Hermitian matrices}

Rather than selecting a distinct value of $\eta_{\rm KR}$ for each matrix instance to meet the target error, we determine a single global value of $\eta_{\rm KR}$ for all instances associated with a fixed condition number. This global $\eta_{\rm KR}$ is chosen such that the average error across all 100 non-Hermitian matrix samples matches the desired target threshold. This strategy is also employed in~\cite{costa2023discrete}, but there in terms of the walk steps: the total number of walk steps is fixed so that, on average over all tested instances, the error matches the prescribed target threshold.

While the primary focus in~\cite{costa2023discrete} was to study the performance of the improved randomization method from~\cite{jennings2023efficient}—particularly in the low condition number regime where it was argued to outperform the discrete adiabatic method—we extend the analysis here to further explore the behaviour of the discrete adiabatic and Randomised method across a broader parameter range; see \cref{sec:further}. Accordingly, for each matrix dimension considered, we compare the complexities of the Shortcut and discrete adiabatic method. 

From the results in \cref{tab:shortcut_only_d04_dim32,tab:shortcut_only_d04_dim64,tab:shortcut_only_d001_dim32,tab:shortcut_only_d001_dim64}, where both choices of allowable error for the solution $\epsilon$ are tested, namely $\epsilon=10^{-2}$ and $\epsilon=10^{-3}$  the Shortcut method on non-Hermitian instances achieves a lower cost than the QW method, \emph{even when counting only the adiabatic component} of the latter (\cref{tab:costNH_32,tab:costNH_64}). Thus, in this idealized scenario, the Shortcut method outperforms the QW approach. We also see that the practical performance on this ensemble of random matrices is slightly better than the worst-case bound from \cite[Eq.~(22)]{dalzell2024shortcut} of $\sim \kappa \ln(2 \sqrt{2}/\epsilon)$ pairs of calls to $U_A$, which evaluates to $5.64\kappa$ and $7.94\kappa$ for $\epsilon = 0.01$ and $\epsilon = 0.001$, respectively. 

\begin{table}[h!]
\centering
\caption{Shortcut method performance for $\epsilon=0.01$ on $A\in\mathbb{R}^{32\times 32}$ (100 non-Hermitian instances). 
For each condition number $\kappa$, the parameter $\eta_{\rm KR}$ is chosen so that the average error is $\le 0.01$.  We report the normalized cost $\mathrm{Cost}_{\mathrm{avg}}/\kappa$ for different $\kappa$ values.}
\label{tab:shortcut_only_d04_dim32}
\renewcommand{\arraystretch}{1.2}
\begin{tabular}{|c| c |c |c|}
\hline
$\kappa$ & $\eta_{\rm KR}$ & $\text{Cost}_{\mathrm{avg}}/\kappa$ & $\text{Error}_{\mathrm{avg}}$ \\
\hline
20     &$0.009$    &$5.58$ & $7.00 \times10^{-3}$ \\
40     &$0.016$    &$5.20$ & $9.20 \times10^{-3}$\\
80     &$0.016$    &$4.95$ & $9.80 \times10^{-3}$  \\
160    &$0.016$    &$4.76$ & $9.60 \times10^{-3}$ \\
320    &$0.025$    &$4.50$ & $1.00 \times10^{-2}$\\
640    &$0.030$    &$4.31$ & $9.90 \times10^{-3}$\\
1280   &$0.005$    &$3.86$ & $9.80 \times10^{-3}$\\
2560   &$0.070$    &$3.38$ & $9.50 \times10^{-3}$\\
\hline
\end{tabular}
\end{table}

\begin{table}[h!]
\centering
\caption{Shortcut method performance for $\epsilon=0.01$ on $A\in\mathbb{R}^{64\times 64}$ (100 non-Hermitian instances). 
For each condition number $\kappa$, the parameter $\eta_{\rm KR}$ is chosen so that the average error is $\le 0.01$. We report the normalized cost $\mathrm{Cost}_{\mathrm{avg}}/\kappa$ for different $\kappa$ values.}
\label{tab:shortcut_only_d04_dim64}
\renewcommand{\arraystretch}{1.2}
\begin{tabular}{|c| c |c |c|}
\hline
$\kappa$ & $\eta_{\rm KR}$ & $\text{Cost}_{\mathrm{avg}}/\kappa$ & $\text{Error}_{\mathrm{avg}}$ \\
\hline
20     &$0.009$    &$5.59$ & $7.60 \times10^{-3}$ \\
40     &$0.010$    &$5.35$ & $9.00 \times10^{-3}$\\
80     &$0.012$    &$5.23$ & $9.70 \times10^{-3}$  \\
160    &$0.012$    &$5.01$ & $9.90 \times10^{-3}$ \\
320    &$0.018$    &$4.84$ & $9.70 \times10^{-3}$\\
640    &$0.023$    &$4.59$ & $9.10 \times10^{-3}$\\
1280   &$0.032$    &$4.25$ & $1.00 \times10^{-2}$\\
2560   &$0.042$    &$3.95$ & $9.20 \times10^{-3}$\\
\hline
\end{tabular}
\end{table}

\begin{table}[h!]
\centering
\caption{Shortcut method performance for $\epsilon=0.001$ on $A\in\mathbb{R}^{32\times 32}$, averaged over 100 non-Hermitian instances. For each condition number $\kappa$, we select $\eta_{\rm KR}$ such that the average error is at most $10^{-3}$. We report the normalized cost $\mathrm{Cost}_{\mathrm{avg}}/\kappa$ for different $\kappa$ values.}
\label{tab:shortcut_only_d001_dim32}
\renewcommand{\arraystretch}{1.2}
\begin{tabular}{|c|c|c|c|}
\hline
$\kappa$ & $\eta_{\rm KR}$ & $\text{Cost}_{\mathrm{avg}}/\kappa$ & $\text{Error}_{\mathrm{avg}}$ \\
\hline
20    & $0.0012$ & $7.52$ & $9.74 \times 10^{-4}$ \\
40    & $0.0014$ & $7.30$ & $9.80 \times 10^{-4}$ \\
80    & $0.0016$ & $7.18$ & $1.00 \times 10^{-3}$ \\
160   & $0.0020$ & $6.96$ & $1.00 \times 10^{-3}$ \\
320   & $0.0025$ & $6.76$ & $9.70 \times 10^{-4}$ \\
640   & $0.0029$ & $6.60$ & $9.35 \times 10^{-4}$ \\
1280  & $0.0050$ & $6.24$ & $9.86 \times 10^{-4}$ \\
2560  & $0.0090$ & $5.54$ & $1.00 \times 10^{-4}$ \\
\hline
\end{tabular}
\end{table}

\begin{table}[h!]
\centering
\caption{Shortcut method performance for $\epsilon=0.001$ on $A\in\mathbb{R}^{64\times 64}$, averaged over 100 non-Hermitian instances. For each condition number $\kappa$, we choose $\eta_{\rm KR}$ such that the average error is at most $10^{-3}$. We report the normalized cost $\mathrm{Cost}_{\mathrm{avg}}/\kappa$ for different $\kappa$ values.}
\label{tab:shortcut_only_d001_dim64}
\renewcommand{\arraystretch}{1.2}
\begin{tabular}{|c|c|c|c|}
\hline
$\kappa$ & $\eta_{\rm KR}$ & $\text{Cost}_{\mathrm{avg}}/\kappa$ & $\text{Error}_{\mathrm{avg}}$ \\
\hline
20   & $0.0011$ & $7.62$ & $9.20 \times 10^{-4}$ \\
40   & $0.0012$ & $7.46$ & $9.96 \times 10^{-4}$ \\
80   & $0.0103$ & $7.36$ & $1.00 \times 10^{-3}$ \\
160  & $0.0015$ & $7.22$ & $1.00 \times 10^{-4}$ \\
320  & $0.0016$ & $7.20$ & $9.70 \times 10^{-4}$ \\
640  & $0.0230$ & $6.82$ & $1.00 \times 10^{-3}$ \\
1280 & $0.0030$ & $6.58$ & $1.00 \times 10^{-3}$ \\
2560 & $0.0056$ & $5.98$ & $1.00 \times 10^{-3}$ \\
\hline
\end{tabular}
\end{table}

\subsubsection{Positive definite matrices}

We next test the Shortcut method on positive-definite matrices with a known solution norm.
For these instances, the QW method is known to have a substantially smaller constant factor than
in the non-Hermitian case—approximately $\alpha \approx 0.17$ for PD versus $\alpha \approx 1.84$
for non-Hermitian,  where $\alpha$ is the constant factor used in the adiabatic part for a given  error $\Delta$, i.e., $T=\alpha \kappa/\Delta$, as reported in~\cite{costa2023discrete}. We therefore benchmark both methods side by side under these conditions.

\begin{table}[tbh]
\centering
\caption{Shortcut method performance for $\epsilon=0.01$ on $A\in\mathbb{R}^{32\times 32}$ (100 PD instances). 
For each condition number $\kappa$, the parameter $\eta_{\rm KR}$ is chosen so that the average error is $\le 0.01$.We report the normalized cost $\mathrm{Cost}_{\mathrm{avg}}/\kappa$ for different $\kappa$ values.}
\label{tab:shortcut_only_d04_dim32PD}
\renewcommand{\arraystretch}{1.2}
\begin{tabular}{|c| c |c |c|}
\hline
$\kappa$ & $\eta_{\rm KR}$ & $\text{Cost}_{\mathrm{avg}}/\kappa$ & $\text{Error}_{\mathrm{avg}}$ \\
\hline
20     &$0.021$    &$4.70$ & $9.90 \times10^{-3}$ \\
40     &$0.025$    &$4.49$ & $9.40 \times10^{-3}$\\
80     &$0.030$    &$4.30$ & $9.30 \times10^{-3}$  \\
160    &$0.040$    &$4.01$ & $9.00 \times10^{-3}$ \\
320    &$0.070$    &$3.49$ & $1.04 \times10^{-2}$\\
640    &$0.110$    &$3.19$ & $1.01 \times10^{-2}$\\
1280   &$0.100$    &$4.23$ & $1.00 \times10^{-2}$\\
2560   &$0.150$    &$3.02$ & $9.40 \times10^{-3}$\\
\hline
\end{tabular}
\end{table}

\begin{table}[tbh]
\centering
\caption{Shortcut method performance for $\epsilon=0.01$ on $A\in\mathbb{R}^{64\times 64}$, averaged over 100 PD instances. For each condition number $\kappa$, we choose $\eta_{\rm KR}$ such that the average error is at most $10^{-2}$. For each condition number $\kappa$, we choose $\eta_{\rm KR}$ such that the average error is at most $10^{-2}$. We report the normalized cost $\mathrm{Cost}_{\mathrm{avg}}/\kappa$ for different $\kappa$ values.}
\label{tab:shortcut_only_d001_dim64PD}
\renewcommand{\arraystretch}{1.2}
\begin{tabular}{|c|c|c|c|}
\hline
$\kappa$ & $\eta_{\rm KR}$ & $\text{Cost}_{\mathrm{avg}}/\kappa$ & $\text{Error}_{\mathrm{avg}}$ \\
\hline
20   & $0.020$ & $4.80$ & $9.40 \times 10^{-3}$ \\
40   & $0.023$ & $4.62$ & $9.50 \times 10^{-3}$ \\
80   & $0.024$ & $4.52$ & $8.90 \times 10^{-3}$ \\
160  & $0.028$ & $4.38$ & $9.80 \times 10^{-3}$ \\
320  & $0.032$ & $4.24$ & $9.90 \times 10^{-3}$ \\
640  & $0.038$ & $4.04$ & $9.20 \times 10^{-3}$ \\
1280 & $0.044$ & $3.90$ & $9.80 \times 10^{-3}$ \\
2560 & $0.065$ & $3.52$ & $1.00 \times 10^{-2}$ \\
\hline
\end{tabular}
\end{table}

\begin{table}[h!]
\centering
\caption{Shortcut method performance for $\epsilon=0.001$ on $A\in\mathbb{R}^{32\times 32}$ (100 PD instances). 
For each condition number $\kappa$, the parameter $\eta_{\rm KR}$ is chosen so that the average error is $\le 0.001$. We report the normalized cost $\mathrm{Cost}_{\mathrm{avg}}/\kappa$ for different $\kappa$ values.}
\label{tab:shortcut_only_d001_dim32PD}
\renewcommand{\arraystretch}{1.2}
\begin{tabular}{|c| c |c |c|}
\hline
$\kappa$ & $\eta_{\rm KR}$ & $\text{Cost}_{\mathrm{avg}}/\kappa$ & $\text{Error}_{\mathrm{avg}}$ \\
\hline
20    &$0.0024$    &$6.80$ & $9.91 \times10^{-4}$ \\
40    &$0.0027$    &$6.65$ & $1.00 \times10^{-3}$\\
80    &$0.0036$    &$6.38$ & $1.00 \times10^{-3}$  \\
160   &$0.0046$    &$6.14$ & $1.00 \times10^{-3}$ \\
320   &$0.0072$    &$5.71$ & $1.00 \times10^{-3}$\\
640   &$0.0084$    &$5.56$ & $9.71 \times10^{-4}$\\
1280  &$0.0120$    &$5.20$ & $9.94 \times10^{-4}$\\
2560  &$0.0170$    &$4.86$ & $1.00 \times10^{-3}$\\
\hline
\end{tabular}
\end{table}

\begin{table}[h!]
\centering
\caption{Shortcut method performance for $\epsilon=0.001$ on $A\in\mathbb{R}^{64\times 64}$ (100 PD instances). 
For each condition number $\kappa$, the parameter $\eta_{\rm KR}$ is chosen so that the average error is $\le 0.001$. We report the normalized cost $\mathrm{Cost}_{\mathrm{avg}}/\kappa$ for different $\kappa$ values.}
\label{tab:shortcut_only_d001_dim64PD2}
\renewcommand{\arraystretch}{1.2}
\begin{tabular}{|c| c |c |c|}
\hline
$\kappa$ & $\eta_{\rm KR}$ & $\text{Cost}_{\mathrm{avg}}/\kappa$ & $\text{Error}_{\mathrm{avg}}$ \\
\hline
20     &$0.0024$    &$6.80$ & $1.00 \times10^{-3}$ \\
40     &$0.0024$    &$6.75$ & $9.89 \times10^{-4}$\\
80     &$0.0028$    &$6.60$ & $1.00 \times10^{-3}$  \\
160    &$0.0032$    &$6.49$ & $1.00 \times10^{-3}$ \\
320    &$0.0036$    &$6.38$ & $9.53 \times10^{-4}$\\
640    &$0.0044$    &$6.19$ & $1.00 \times10^{-3}$\\
1280   &$0.0066$    &$5.80$ & $9.87 \times10^{-4}$\\
2560   &$0.0082$    &$5.58$ & $1.00 \times10^{-3}$\\
\hline
\end{tabular}
\end{table}

\begin{figure}[H]
    \centering
\includegraphics[width=0.45\textwidth]{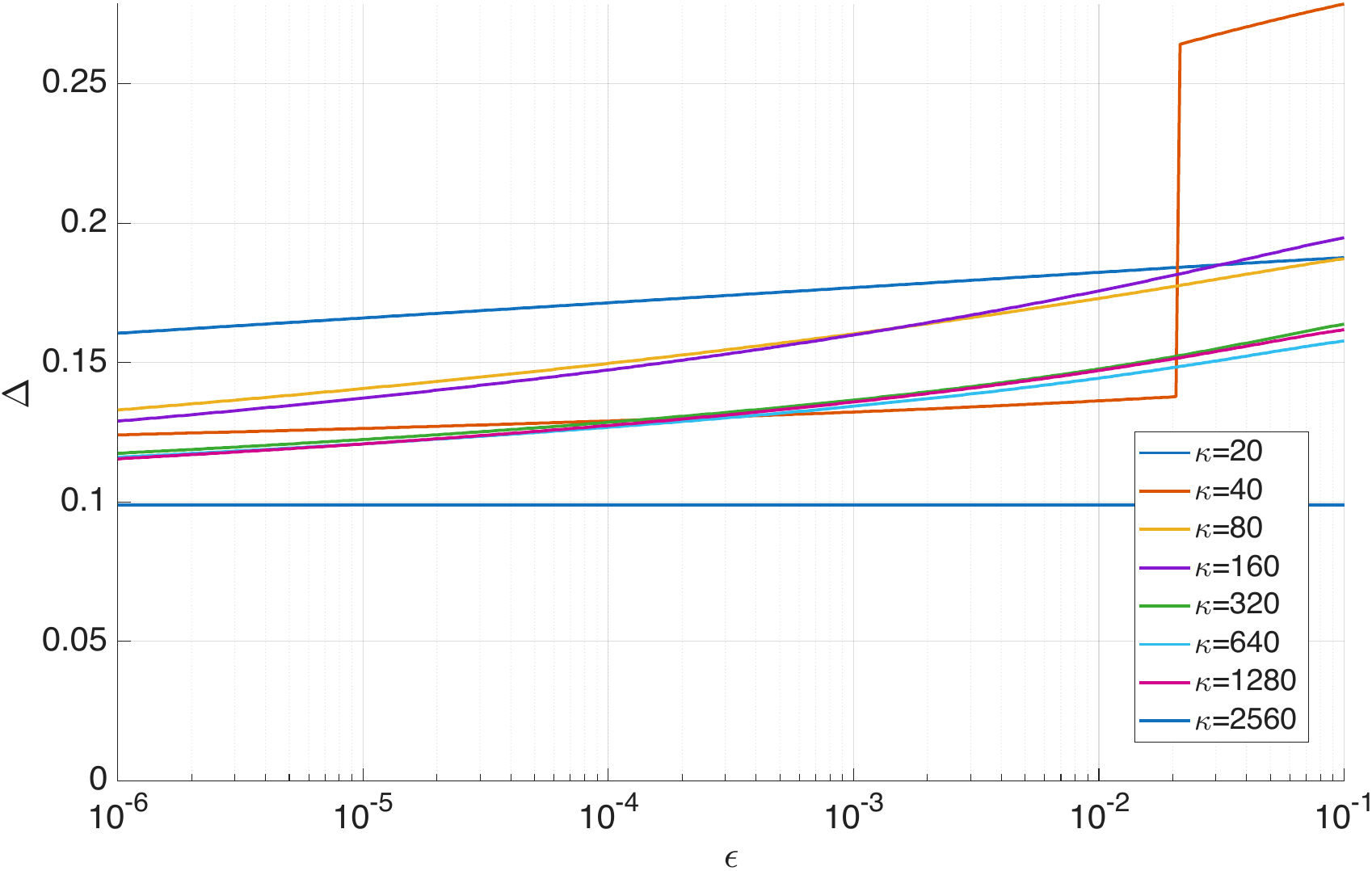}
    \caption{The plot shows the recommended values of $\Delta$, for the QW method, as a function of the tested condition number, for the set of positive definite matrices $A\in\mathbb{C}^{32\times 32}$.}
    \label{fig:Delta32}
\end{figure}

\begin{figure}[H]
    \centering
\includegraphics[width=0.45\textwidth]{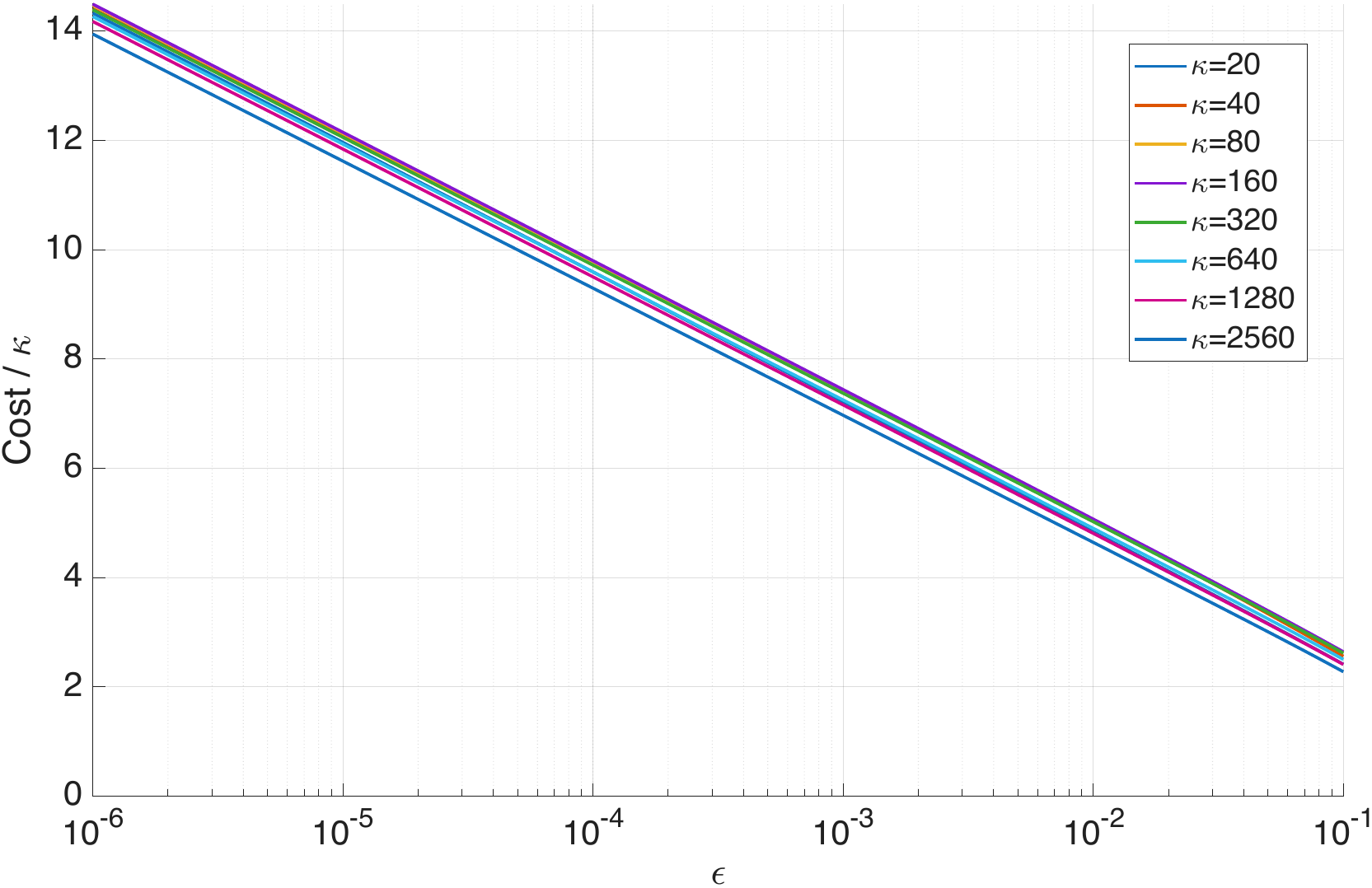}
    \caption{Total cost of the QW method for different target solution-error tolerances, using the optimal step sizes $\Delta$ (as estimated from \cref{fig:Delta32}) for each tested condition number, over the set of positive definite matrices $A\in\mathbb{C}^{32\times 32}$.}
    \label{fig:Cost32}
\end{figure}

\begin{figure}[H]
    \centering
\includegraphics[width=0.45\textwidth]{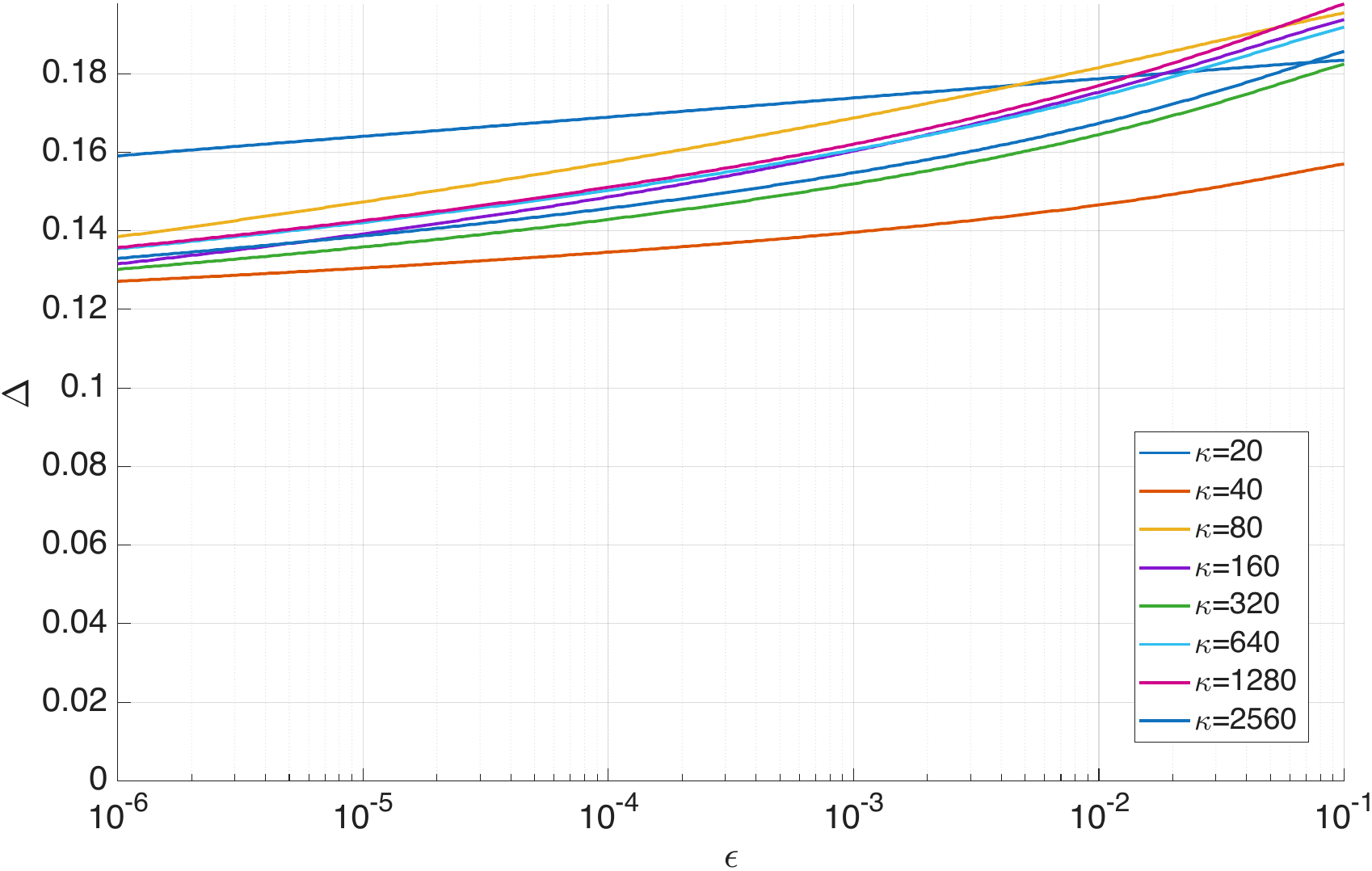}
    \caption{The plot shows the recommended values of $\Delta$, for the QW method, as a function of the tested condition number, for the set of positive definite matrices $A\in\mathbb{C}^{64\times 64}$.}
    \label{fig:Delta64}
\end{figure}

\begin{figure}[H]
    \centering
\includegraphics[width=0.45\textwidth]{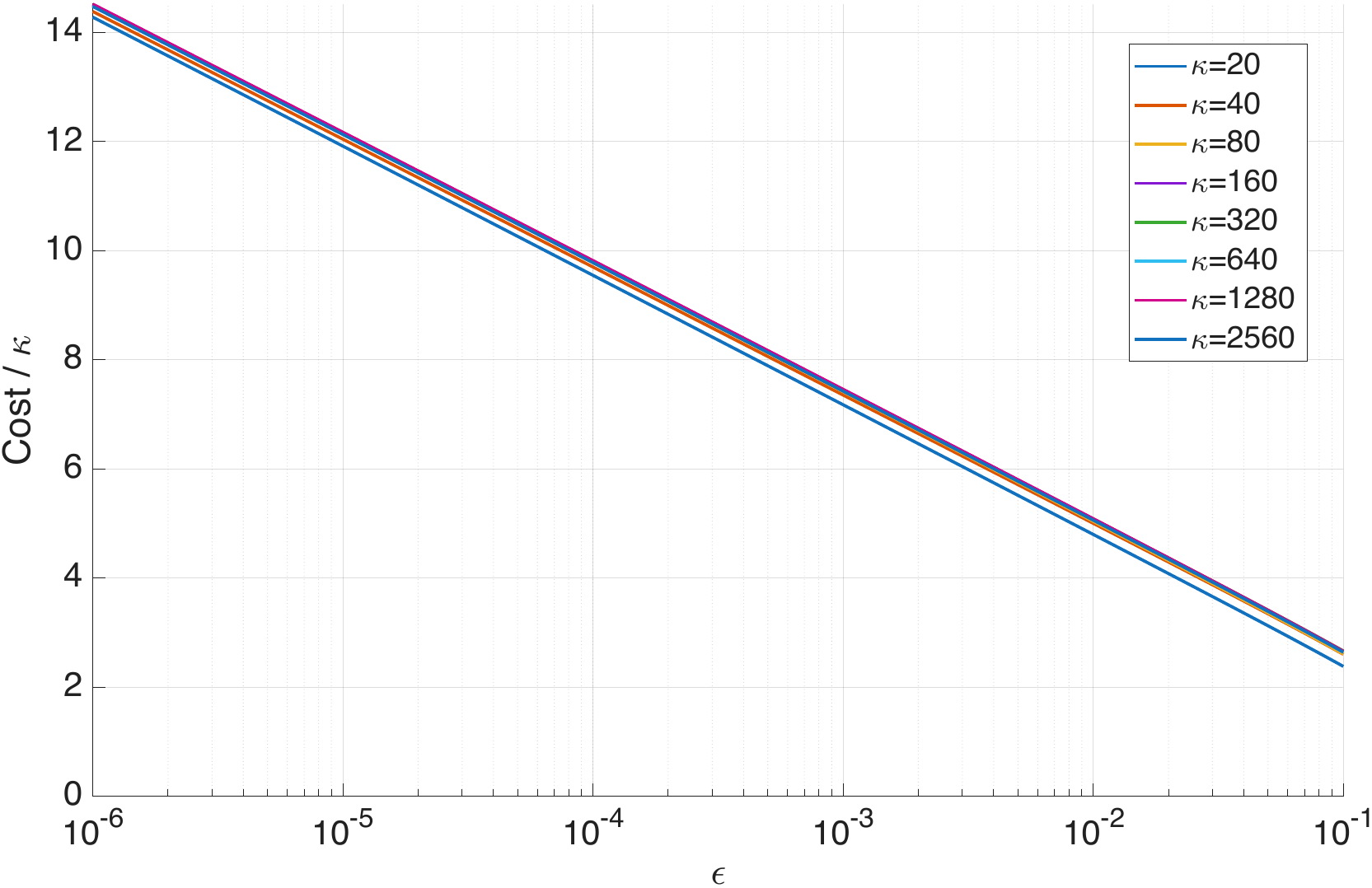}
    \caption{Total cost of the QW method for different target solution-error tolerances, using the optimal step sizes $\Delta$ (as estimated from \cref{fig:Delta64}) for each tested condition number, over the set of positive definite matrices $A\in\mathbb{C}^{64\times 64}$.}
    \label{fig:Cost64}
\end{figure}

\clearpage

The results for the adiabatic part, for different values of $\Delta$, for the same matrices tested with the Shortcut method are presented in \cref{tab:cost_PD_32,tab:cost_PD_64}. We then used the results to estimate the optimal values of $\Delta$ that minimise the total cost (adiabatic plus filtering) for a given total allowable error $\epsilon$. 

The results indicate that regardless of the dimension considered, with the optimal values \cref{fig:Delta32} and \cref{fig:Delta64}, the total cost of the QW method, \cref{fig:Cost32} and \cref{fig:Cost64} (see the ratios given for $10^{-2}$ and $10^{-3}$) for the PD almost matches the results for the Shortcut when the norm of the solution is known, with a slight improvements for the Shortcut method for the largest condition numbers, namely $\kappa=1280,2560$. However, in the case that we count the QW step as a twice number of steps in the adiabatic part, the shortcut method will outperform the QW method; see \cref{fig:Cost32double,fig:Cost64double} for $10^{-2}$ and $10^{-3}$.  

\subsection{Performance in the Unknown-Norm Regime}\label{sec:unk_regime}

We also tested the Shortcut method for non-Hermitian matrices when the norm is unknown for $32\times 32$ and $64\times64$ non-Hermitian matrices, but only up to $\kappa=320$. In this situation, the Shortcut method also requires the filtering step, provided by the KP. Specifically, we simulated  Algorithm 3 of  \cite{dalzell2024shortcut}, which is a variant of the ``exhaustive search in log space'' described in Section 5.1 of \cite{dalzell2024shortcut}. The idea there is to make random guesses for the norm, and try running the KR algorithm with a fixed choice of $\eta_{\rm KR}$. This succeeds some of the time and produces an output state with error $\Delta$ conditioned on success. This output state is subsequently filtered to error $\epsilon$.
Rather than take the total cost including the filter based on the KP as proposed in \cite{dalzell2024shortcut}, we consider the cost based on the LCU filtering method from \cite{costa2022optimal}.
The motivation is that, if there is failure (i.e., the state is measured as orthogonal to the solution), then it will typically be detected early.
On average, a failure is detected halfway through the LCU, halving the average cost, whereas KP needs to be complete before a failure is detected.
For further explanation, see \cref{app:QWfilt} or Section 4 of \cite{costa2023discrete}.

Before presenting the numerical results, we briefly describe how we test the Shortcut method in the unknown-norm regime. The method to calculate the infidelity is given in Appendix D of Ref.~\cite{dalzell2024shortcut}.
Equation (85) of Ref.~\cite{dalzell2024shortcut} gives an expression for the success probability averaged over random choice of $t$:
\begin{align}
    q_{\rm succ} 
    &= \frac 1{2\ln(\mathcal{R}/\mathcal{L})+2}\left( \mathcal{Q}_\mathcal{L} + \mathcal{Q}_\mathcal{R} + 2\int_\mathcal{L}^\mathcal{R} d\tau \, \mathcal{Q}_t  \right),
\end{align}
where $t=e^\tau$, $\mathcal{Q}_t$ is the probability of success for a given $t$, $\mathcal{L}=1$ and $\mathcal{R}=\kappa$.
The infidelity can then be determined using Eq.~(98) of Ref.~\cite{dalzell2024shortcut}, which is
\begin{align}
    1-\langle\mathbf{x}|\tilde\rho|\mathbf{x}\rangle &= \frac{1}{q_{\rm succ}(2\ln(\mathcal{R}/\mathcal{L})+2)}\nn
    &\quad \times
    \left( \mathcal{Q}_\mathcal{L}\mu_\mathcal{L}^2 + \mathcal{Q}_\mathcal{R}\mu_\mathcal{R}^2 + 2\int_\mathcal{L}^\mathcal{R} d\tau \, \mathcal{Q}_t  \mu_t^2 \right),
\end{align}
where $\mu_t^2$ is the infidelity for the particular choice of $t$.

In that work, upper bounds are determined analytically, but in numerical testing we perform the integral numerically by calculating $\mathcal{Q}_t$ and $\mu_t$ for selected values of $t$.
We used Clenshaw–Curtis quadrature, which provides highly accurate results with only 30 samples of $t$.
As this method considers a mixed state, $\Delta$ will be the Bures distance, which can be determined from the infidelity $\mu^2$ using
\begin{equation}
    \Delta = \sqrt{2\left(1-\sqrt{1-\mu^2}\right)},
\end{equation}
which is the same as the formula for the norm of the difference of pure states. We also record the cost of the Shortcut method, defined as the degree of the polynomial required for the chosen value of~$\eta_{\rm KR}$. In our tests, $\eta_{\rm KR}$ is fixed for each matrix instance, so the cost is constant across guesses for that instance. For each of the 100 random matrix instances, we compute both the error~$\Delta_i$ and the corresponding cost, and then average these per-instance quantities over the ensemble. The value of~$\eta_{\rm KR}$ is chosen so that the mean error remains below the prescribed threshold~$\Delta$.

We then used the recommended values of $\Delta$ (see \cref{app:recom_delta}) to estimate the total cost for the Shortcut method, reported in \cref{fig:Cost_Sht32}. The results in \cref{fig:Cost_Sht32} may be compared to Eq.~(128) of \cite{dalzell2024shortcut}, which rigorously upper bounds the cost as (after some simplifications)
\begin{align}
    \frac{\mathrm{cost}}{\kappa} \leq 6\ln(\kappa) + 6 + 1.07 \ln(1/\epsilon) + (6+3\ln(\kappa))/\kappa\,.
\end{align}
As expected, the numerical results averaged over the instances exhibit lower constant factors than the rigorous worst-case analysis. As one example,  at $\kappa = 320$ and $\epsilon = 10^{-3}$ the analytic upper bound gives $\mathrm{cost}/\kappa \approx 48$, which is about than 60\% larger than the numerically computed value.

\begin{figure}[H]
    \centering
\includegraphics[width=0.45\textwidth]{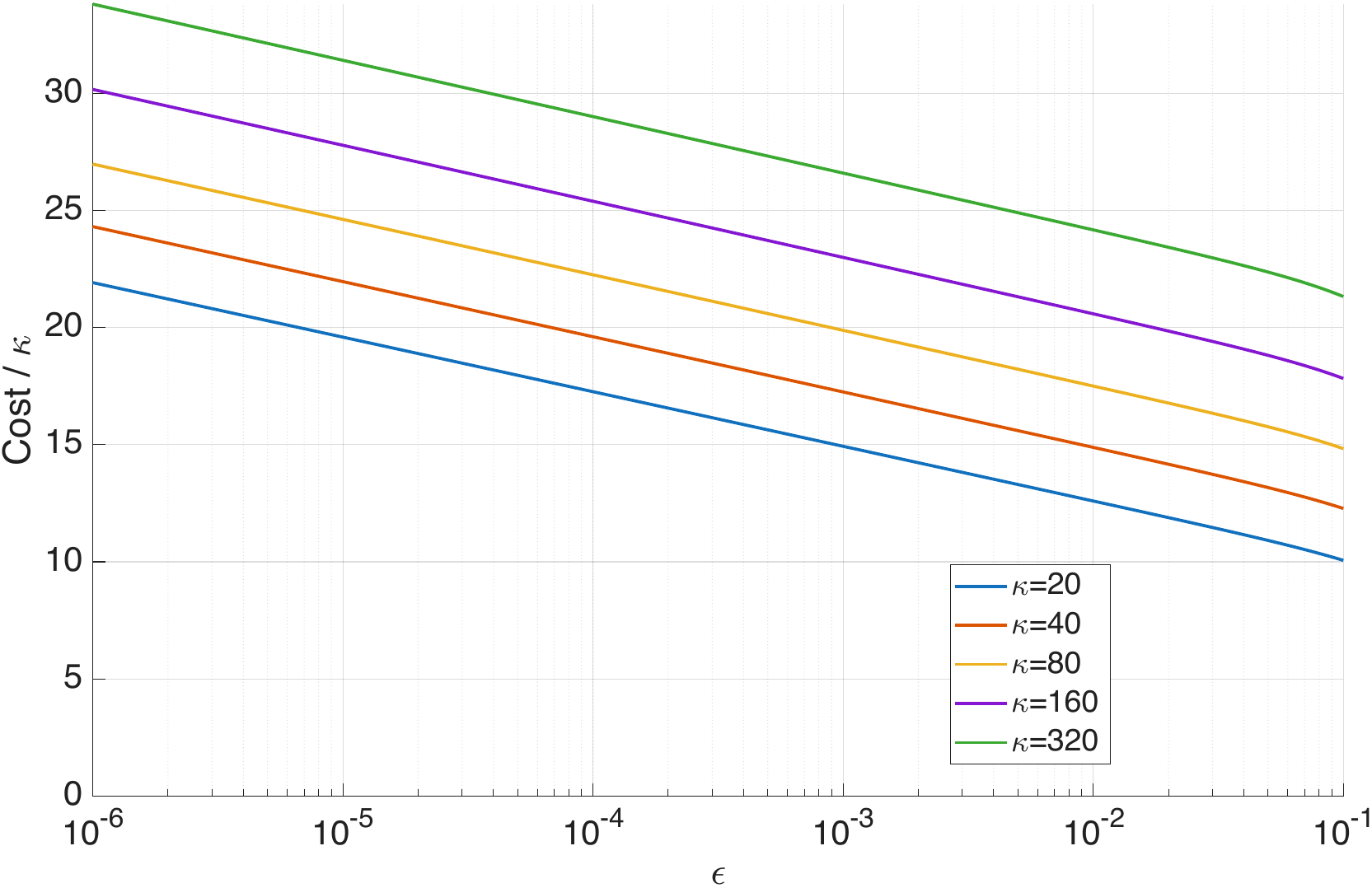}
    \caption{Total cost for the Shortcut method for different target solution-error tolerances, using the optimal target error $\Delta$ (as estimated from \cref{fig:recomenSH_delta32}) for each tested condition number, over the set of non-Hermitian matrices $A\in\mathbb{C}^{32\times 32}$.}
    \label{fig:Cost_Sht32}
\end{figure}

We then used the recommended values of $\Delta$ (see \cref{app:recom_delta})  to estimate the total cost for the QW method, \cref{fig:Cost_QW_NH32}, for the same non-Hermitian matrices tested for the Shortcut method.

\begin{figure}[H]
    \centering
\includegraphics[width=0.45\textwidth]{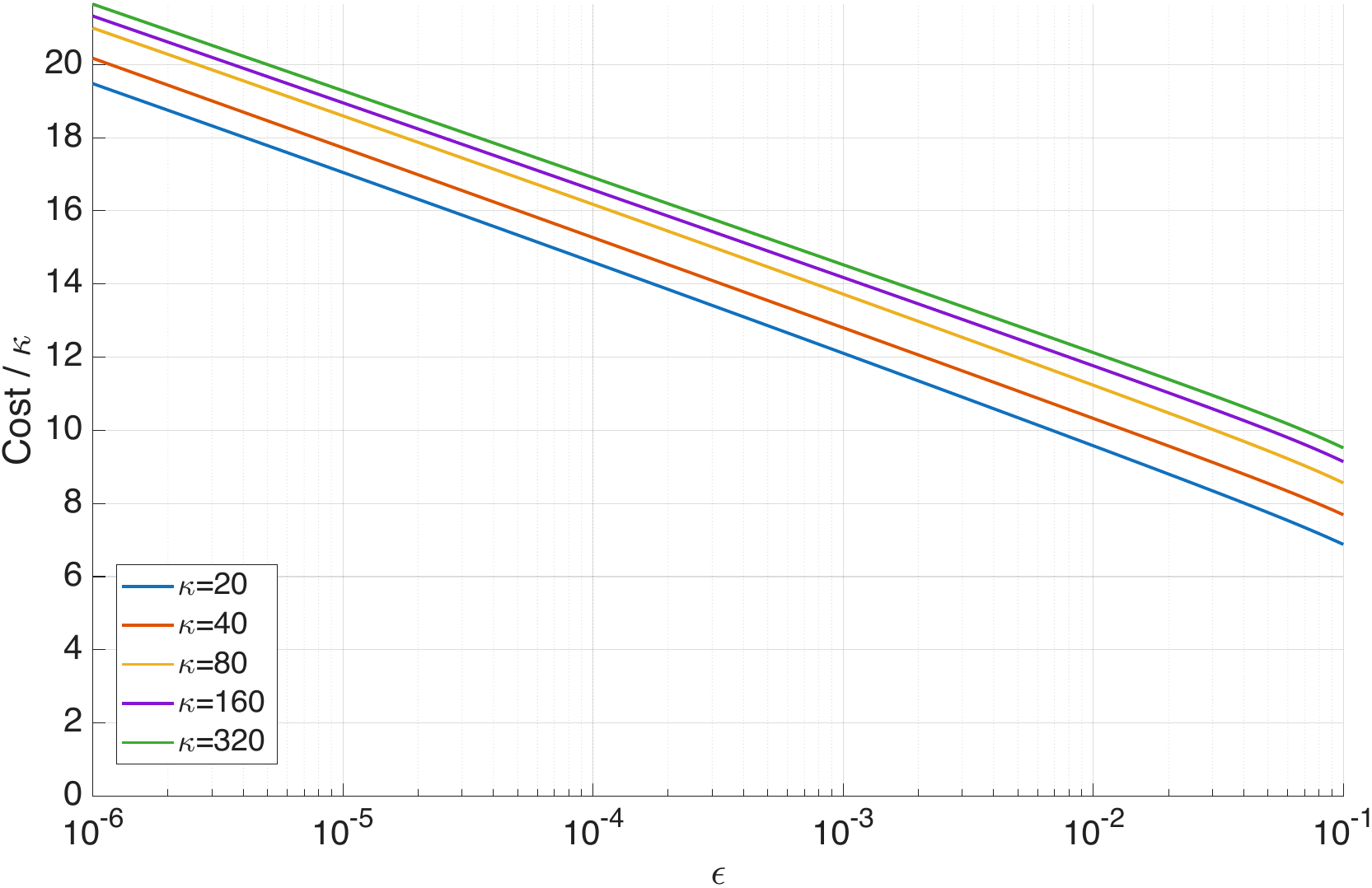}
    \caption{Total cost for the QW method for different target solution-error tolerances, using the optimal step sizes $\Delta$ (as estimated from \cref{fig:delta_NH32_QW}) for each tested condition number, over the set of non-Hermitian matrices $A\in\mathbb{C}^{32\times 32}$.}
    \label{fig:Cost_QW_NH32}
\end{figure}

From our $32\times 32$ testing, we observe that for practically relevant choices of the total allowable solution error $\epsilon$, the QW method  can outperform the Shortcut method (with exhaustive search for the norm) reaching at most a factor of approximately $2.3$ for a target with low accuracy $\epsilon$ (ratios for $\epsilon=10^{-1}$ in \cref{fig:Cost_Sht32} vs \cref{fig:Cost_QW_NH32}) for $\kappa=320$, however reducing it to a factor of $1.4$ in the regime of high accuracy for $\epsilon$ (ratio for $\epsilon=10^{-6}$). On the other hand, for the smallest value of the condition number $\kappa=20$, the cost for both is equivalent.

We next analyze the results for $64\times 64$ matrices, for both methods, following the recommended values for $\Delta$ (see \cref{app:recom_delta}).  Overall, the total costs of both the QW and Shortcut methods remain close to those observed in the $32\times 32$ case. Notably, we observe a slight cost improvement for the Shortcut method (see \cref{fig:Cost_Sh_NH64}) compared to the $32\times 32$ setting, while the QW method exhibits almost the same cost (see \cref{fig:Cost_QW_NH64}), but overall the QW remains as the lowest cost method. On the other hand, if we double the walk steps in the adiabatic method, the Shortcut method performs slightly better than the QW for the lowest condition numbers, namely $\kappa=20$ and $\kappa=40$, see \cref{fig:Cost_QW_NH32double,fig:Cost_QW_NH64double}.

\begin{figure}[H]
    \centering
\includegraphics[width=0.45\textwidth]{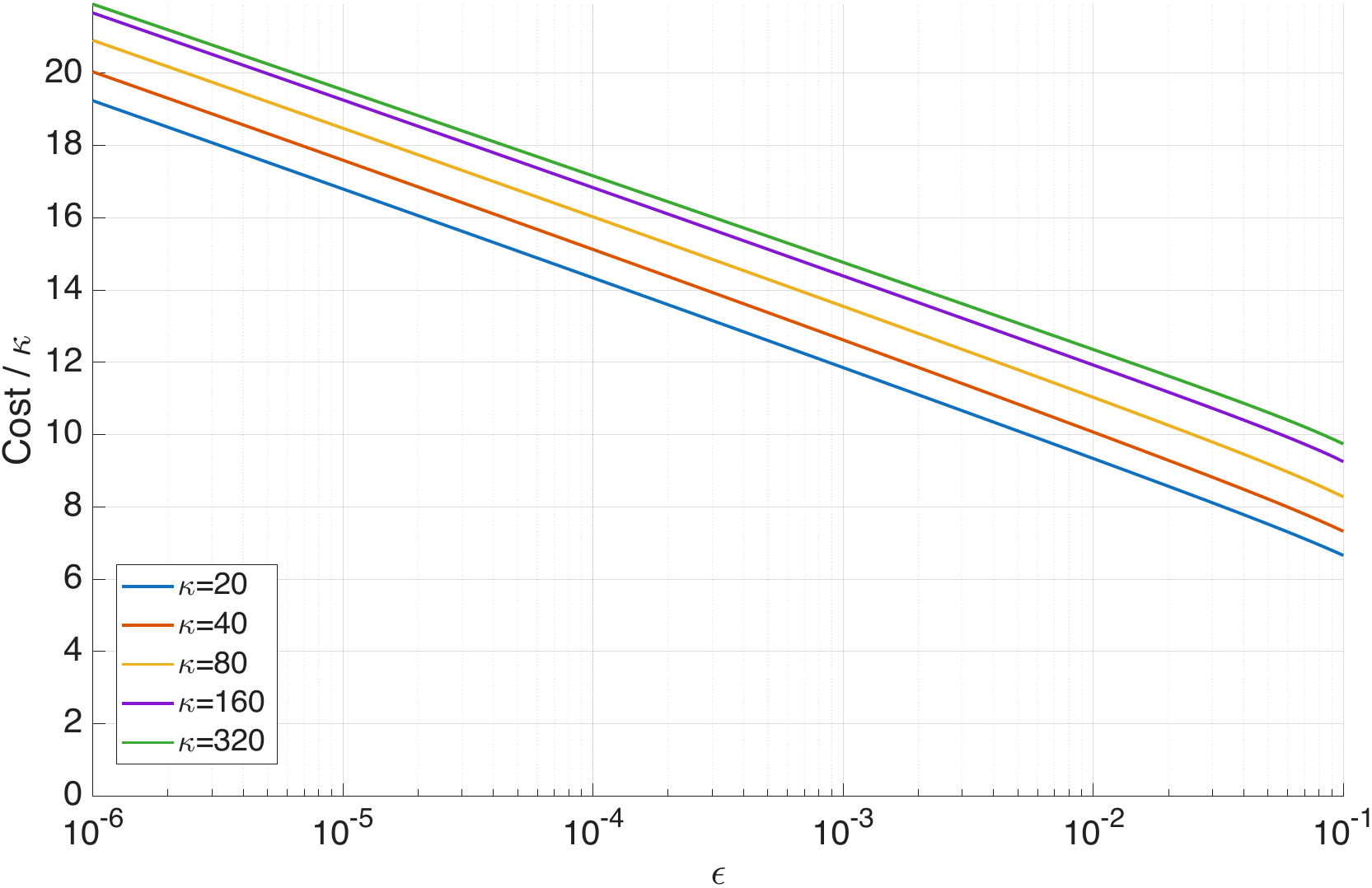}
    \caption{Total cost for the QW method for different target solution-error tolerances, using the optimal step sizes $\Delta$ (as estimated from \cref{fig:delta_NH64_QW}) for each tested condition number, over the set of non-Hermitian matrices $A\in\mathbb{C}^{64\times 64}$.}
    \label{fig:Cost_QW_NH64}
\end{figure}

\begin{figure}[H]
    \centering
\includegraphics[width=0.45\textwidth]{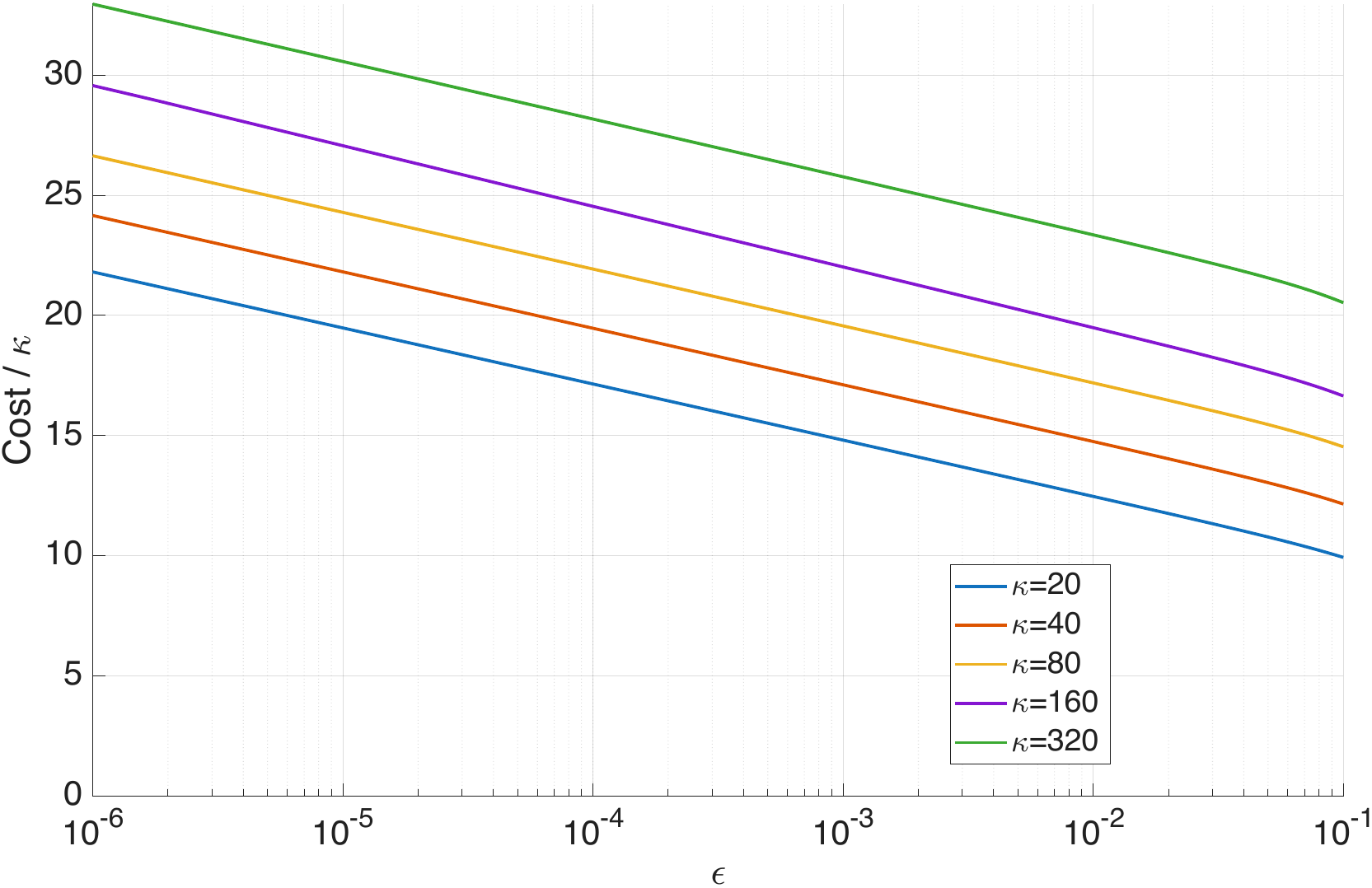}
    \caption{Total cost for the Shortcut method for different target solution-error tolerances, using the optimal step sizes $\Delta$ (as estimated from \cref{fig:delta_NH64_Sh}) for each tested condition number, over the set of non-Hermitian matrices $A\in\mathbb{C}^{64\times 64}$.}
    \label{fig:Cost_Sh_NH64}
\end{figure}

\paragraph{Sparse matrix numerical testing.}

We also benchmarked sparse instances to contrast with the artificial dense cases considered earlier. Following \cite{Costa_QLSP_Shortcut_2025}, the sparse test matrices $A$ were generated using PDE-like routines: we began with a random tridiagonal (Laplacian-like) stencil and, optionally, added a sparse perturbation with user-controlled density. Figure~\ref{fig:sparsity} displays representative sparsity patterns for the tested instances.

We compared the two QLSP solvers (Shortcut and QW) at target error tolerances \(\Delta\in\{0.4,\,0.3\}\) for matrix sizes \(n\in\{32,\,64\}\). Because the generator does not directly prescribe the condition number, we drew \(100\) instances and grouped them into three condition-number overlapping sets:
\begin{align}
&\kappa_1\in[10,30],\quad  (\kappa_1)^{(32\times32)}_{\rm avg}=17.65,\quad (\kappa_1)^{(64\times64)}_{avg}=14.27\nn 
&\kappa_2\in[20,40],\quad  (\kappa_2)^{(32\times32)}_{\rm avg}=26.93,\quad (\kappa_2)^{(64\times64)}_{avg}=26.98\nn
&\kappa_3\in[30,50],\quad(\kappa_3)^{(32\times32)}_{\rm avg}=38.92,\quad (\kappa_3)^{(64\times64)}_{avg}=38.06.    
\end{align}

\begin{figure}[h!]
    \centering
    \subfigimg[width=0.7\linewidth]{(a)}{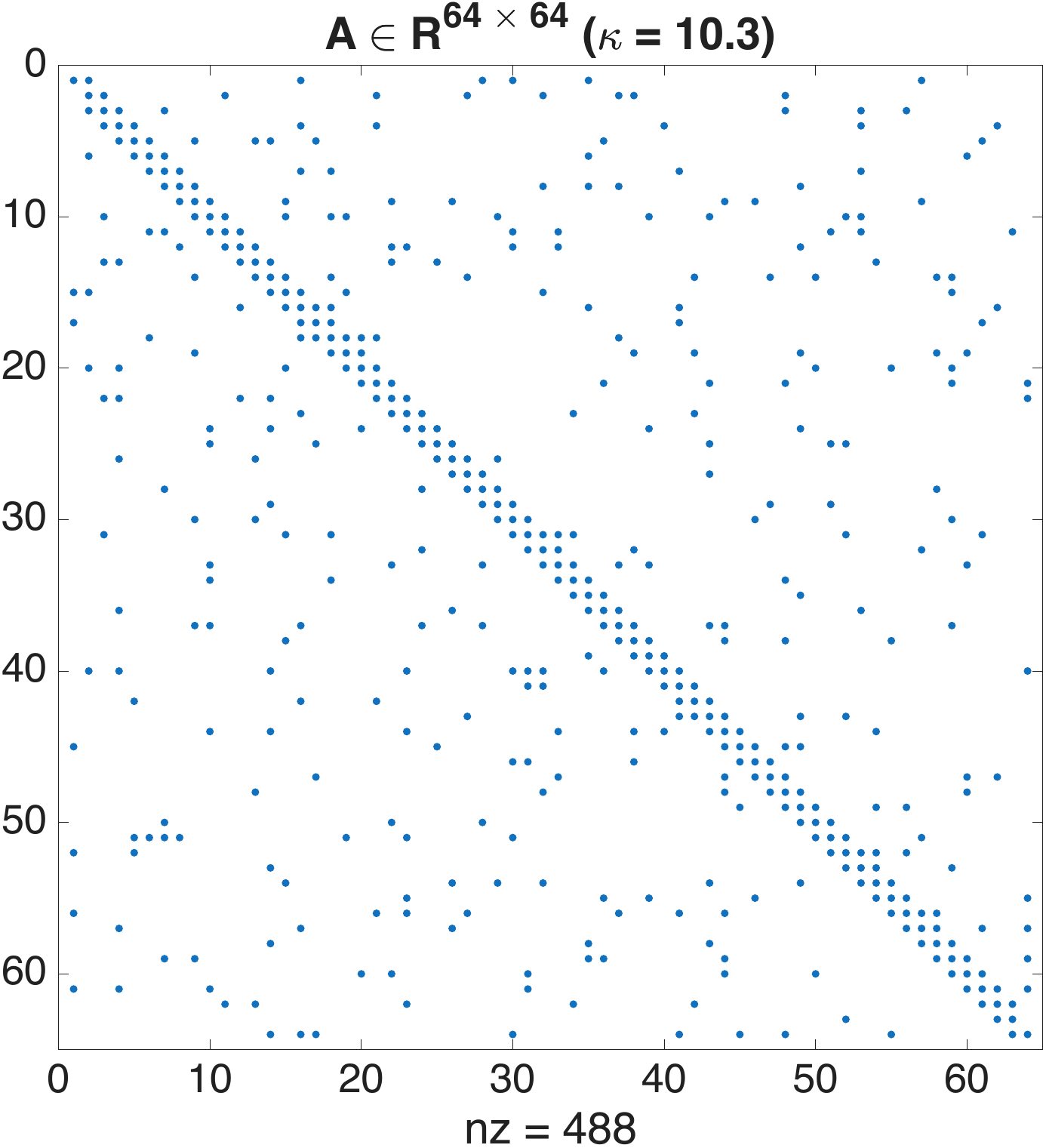} \\ [1em]
    \subfigimg[width=0.7\linewidth]{(b)}{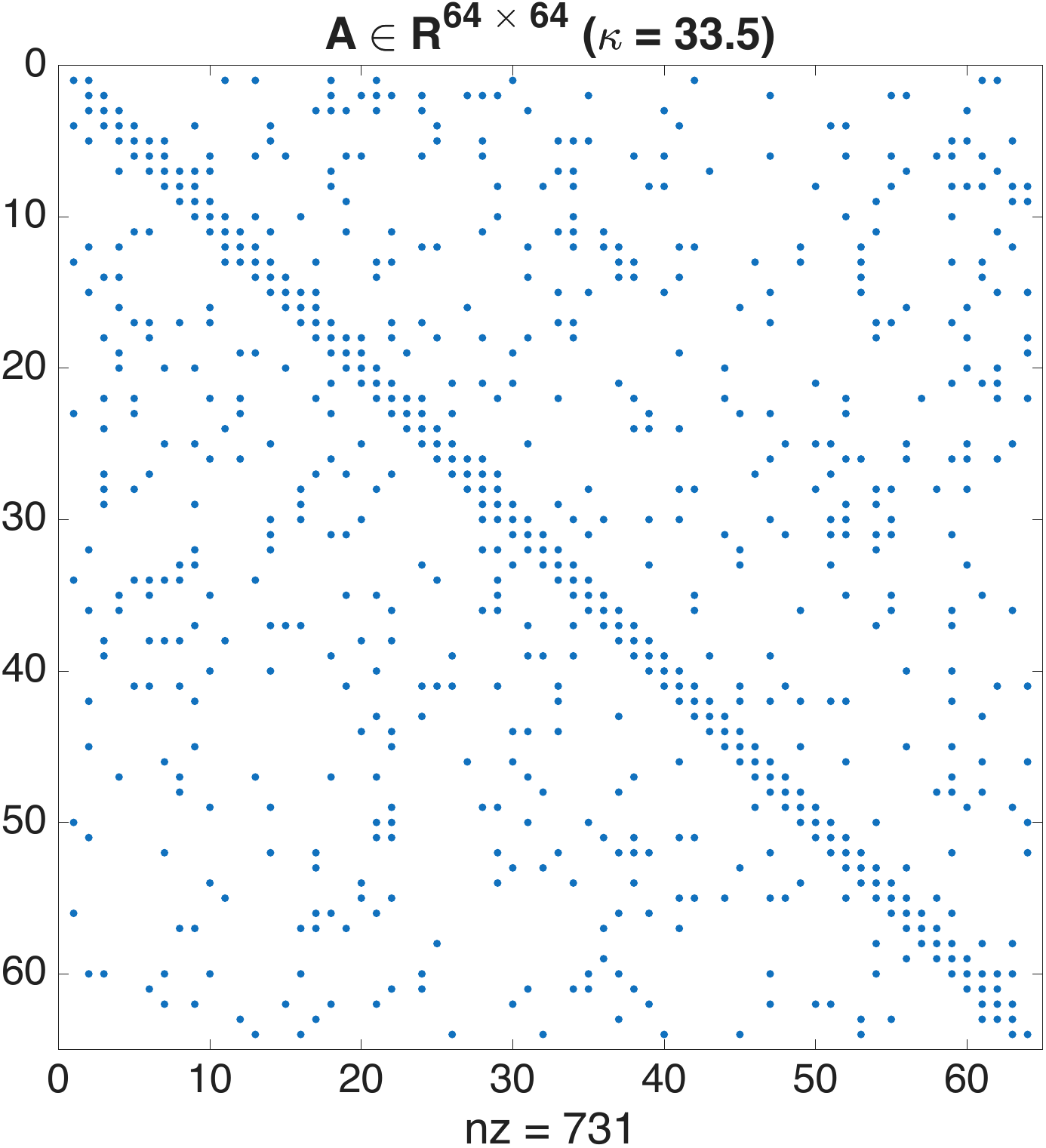} 
    \caption{Two representative $64\times 64$ matrices from our sparse solver testing: (a) a sample from group $k_1$ and (b) a sample from group $\kappa_2$. For each matrix, the top label reports its condition number, and the bottom label reports the total number of nonzeros.}
    \label{fig:sparsity}
\end{figure}

The goal of these tests is to evaluate whether \emph{sparsity}—together with the overheads of block-encoding access to sparse matrices—alters the overall cost beyond the predicted linear dependence on \(\kappa\) and on the target accuracy. The Average cost of the Shortcut method already includes the overhead of estimating the solution norm.

We can analyse these results by examining the cost ratio, $\rho = \text{Cost}_{\mathrm{QW}} / \text{Cost}_{\mathrm{SC}}$. Across the range of conditions considered, we have not observed a significant change for the Shortcut method as compared to what is given in \cref{sec:unk_regime} for the full dense regime. The same conclusion can be observed for the quantum walk method as we can see in \cref{sec:further}. If the cost of a quantum walk step were instead counted as two calls to $U_A$, as discussed in \cref{Sec:result}, the ratio $\rho$ would increase by a factor of 2.

\begin{table}[H]
\centering
\label{tab:performance_sparse_32_delta03}
\renewcommand{\arraystretch}{1.2}
\begin{tabular}{|c|c|c|c|c|c|}
\hline
\multicolumn{6}{|c|}{\textbf{Results for } $A \in \mathbb{R}^{32\times 32}$} \\
\hline
 & \multicolumn{2}{c|}{Shortcut method} & \multicolumn{2}{c|}{Quantum Walk} &{}\\
\hline
$\kappa$
& $\text{(Cost)}_{\mathrm{avg}}$ & $\text{(Error)}_{\mathrm{avg}}$
& $\text{Cost}$ & $\text{(Error)}_{\mathrm{avg}}$
& $\rho$ \\
\hline
$\kappa_1$ & $1.23\times 10^{2}$ & $0.295\pm 0.020$ & $9.14\times 10$ & $0.288\pm 0.001$ & $0.74$ \\
$\kappa_2$ & $2.12\times 10^{2}$ & $0.298\pm 0.011$ & $1.17\times 10^{2}$ & $0.289\pm 0.011$ & $0.55$ \\
$\kappa_3$ & $3.51\times 10^{2}$ & $0.298\pm 0.007$ & $1.96\times 10^{2}$ & $0.293\pm 0.001$ & $0.56$ \\
\hline
\end{tabular}
\caption{Comparison between the Shortcut and Quantum Walk solvers over 100 non-Hermitian sparse instances of size \(32\times 32\) at target error \(\Delta=0.3\). The last column reports the cost ratio \(\rho=\text{Cost}_{\mathrm{QW}}/\text{Cost}_{\mathrm{SC}}\).}
\end{table}

\begin{table}[H]
\centering
\label{tab:performance_sparse_64_delta03}
\renewcommand{\arraystretch}{1.2}
\begin{tabular}{|c|c|c|c|c|c|}
\hline
\multicolumn{6}{|c|}{\textbf{Results for } $A \in \mathbb{R}^{64\times 64}$} \\
\hline
 & \multicolumn{2}{c|}{Shortcut method} & \multicolumn{2}{c|}{Quantum Walk} & {} \\
\hline
$\kappa$
& $\text{(Cost)}_{\mathrm{avg}}$ & $\text{(Error)}_{\mathrm{avg}}$
& $\text{Cost}$ & $\text{(Error)}_{\mathrm{avg}}$
& $\rho$ \\
\hline
$\kappa_1$ & $8.18\times 10$ & $0.300\pm 0.016$ & $4.47\times 10$ & $0.279\pm 0.014$ & $0.55$ \\
$\kappa_2$ & $2.14\times 10^2$ & $0.300\pm 0.010$ & $1.20\times 10^2$ & $0.293\pm 0.006$ & $0.56$ \\
$\kappa_3$ & $3.26 \times 10^2$ & $0.283\pm 0.013$ & $1.64\times 10^2$ & $0.298\pm 0.010$ & $0.50$ \\
\hline
\end{tabular}
\caption{Comparison between the Shortcut and Quantum Walk solvers over 100 non-Hermitian sparse instances of size \(64\times 64\) at target error \(\Delta=0.3\).}
\end{table}

\section{Conclusion}\label{sec:conc}

We tested the Shortcut method~\cite{dalzell2024shortcut} through a large-scale constant-factor benchmarking study, and compared its performance to the discrete-adiabatic quantum-walk solver~\cite{costa2023discrete}.
As compared to the prior benchmarking study Ref.~\cite{costa2023discrete} (which did not consider the Shortcut method), we extended our empirical analysis to larger dimensions, larger condition numbers, and additional matrix families.

Our results confirm that the Shortcut method can provide a clear constant-factor advantage in the idealised regime where $\|\mathbf{x}\|$ is known \emph{a priori} for non-Hermitian matrices. In contrast, for positive-definite systems, the QW method almost matches with the Shortcut performance, consistent with the most favourable constant factors of the QW in this setting. 

In the more realistic unknown-norm regime, however, the overhead associated with norm estimation and filtering renders the Shortcut method significantly more costly than the QW method. As a result, in the large-condition-number and dimensions regime considered here, the QW method shows a clear advantage, making the Shortcut method less competitive overall. Note that this conclusion applies to the simplest variant of the Shortcut method.  It remains possible that optimized implementations of the more complex variants could lead to slightly better performance.

In summary, these findings demonstrate that the relative performance of optimal QLSAs depends sensitively on matrix structure (non-Hermitian vs.\ PD), condition number and the availability of norm information. They provide practical guidance on which QLSAs to choose in various scenarios to provide the best performance.

\subsection*{Acknowledgements}
The authors thank Mauro E. S. Morales for helpful discussions. 
DWB worked on this project under a sponsored research agreement with Google Quantum AI.
This project is supported by Australian Research Council Discovery Projects DP220101602 and DP260102543.
DA acknowledges funding from Quantum Science and Technology - National Science and Technology Major Project via Project 2024ZD0301900, and the support by The Fundamental Research Funds for the Central Universities, Peking University.

\appendix
\counterwithin{figure}{section}
\counterwithin{table}{section}
\renewcommand{\thefigure}{\thesection.\Roman{figure}}
\renewcommand{\thetable}{\thesection.\Roman{table}}

\section{The kernel reflection and the kernel projection}\label{sec:kernels}

The algorithm that we will be focusing on here is based on the kernel projection (KP) and kernel reflection (KR). The KP is the technique of eigenstate filtering, generalized to non-Hermitian matrices, and it leads to approximate projection onto the kernel of a matrix. Now, the KR leads to approximate reflection about the kernel. 

In either case, the first step is to construct a matrix $G \in \mathbb{C}^{m\times m}$ out of $A$ and $\mathbf{b}$, namely
\begin{equation}
G = Q_{\mathbf{b}}A,    
\end{equation}
where $Q_b=I_m - \mathbf{b}\mathbf{b}^{\dagger}$ is the projector onto the subspace of $\mathbf{b}$, which has $\mathbf{x}$ as in the kernel,
\begin{equation}
G\mathbf{x} =  Q_{\mathbf{b}}A\mathbf{x} = Q_{\mathbf{b}} \mathbf{b} = \mathbf{b} - \mathbf{b} \|\mathbf{b}\| = 0,   
\end{equation}
since $\|\mathbf{b}\|=1$. We then need to have the block-encoding of $G$, $U_G$ which has a quite simple construction from the block-encoding of $A$ and the unitary for $U_\mathbf{b}$, as we can see in \cref{fig:UG}, plus an extra operation given by
\begin{equation}
\label{eq:CpNot}
C_{\Pi_0}\text{NOT} = \Pi_{0}\otimes X + (I_{2^s} - \Pi_0)\otimes I_2,    
\end{equation}
where $\Pi_{0}=\ketbra{e_0}$ and $I_{2^{s}} = \sum_{j=0}^{2^s - 1}\ketbra{e_j}$.

\begin{figure}[!ht]
\centerline{
\Qcircuit @R=2em @C=1em {
&\qw& \qw &\multigate{2}{U_G} & \qw &  \\
& {/}^a\qw& \qw & \ghost{U_G}  &\qw & = \\
 & {/}^s\qw& \qw & \ghost{U_G}  &\qw &  
}
\quad\Qcircuit @R=1.6em @C=1em {
 & \qw &\qw & \qw & \targ{} &\qw& \qw\\
  & \qw & \multigate{1}{U_A} & \qw & \qw &\qw&\qw \\
 &\qw & \ghost{U_A} & \gate{U_\mathbf{b}^{\dagger}}& \gate{\ketbra{e_0}}\qwx[-2]& \gate{U_\mathbf{b}}&\qw
}}
\caption{\label{fig:UG}  Quantum circuit implementing the block-enconding of $G$.}
\end{figure}
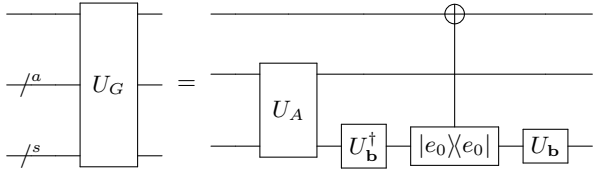

We can easily verify that in the third register we will have exactly the action of $G$, i.e., suppose we start with the following state $\ket{0}\ket{0^a}\ket{\psi}$, where the application of $U_A$, gives
\begin{equation}
I_2 \otimes U_A \ket{0}\ket{0^a}\ket{\psi} = \ket{0}\ket{0^a}A\ket{\psi} + \ket{\phi_{\perp}}.
\end{equation}
Now from \cref{eq:CpNot}, combined with the unitaries $U^{\dagger}_{\mathbf{b}}$ and $U_{\mathbf{b}}$ on the last register, i.e.,
\begin{equation}
 I_2\otimes I_{2^a}\otimes\left(I_{2^s} - \ketbra{\mathbf{b}}\right) + X \otimes I_{2^a}\otimes \ketbra{\mathbf{b}},
\end{equation}
then we finish with
\begin{equation}
\ket{0}\ket{0^a}Q_{\mathbf{b}}A\ket{\psi} + \ket{\phi'_{\perp}} = \ket{0}\ket{0^a}G\ket{\psi} + \ket{\phi'_{\perp}},
\end{equation}
as we wanted.

We can now verify that the smallest nonzero eigenvalue of $G$ is at least $\kappa^{-1}$. Suppose that from the SVD of $G$, i.e., $G=U\Sigma V$, we have $\delta$ a nonzero singular value of $G$, where $G\mathbf{v}=\delta\mathbf{u}$. We then have  $\mathbf{x}$ orthogonal to $\mathbf{v}$, since
\begin{equation}
\delta^*(\mathbf{v}^{\dagger}\mathbf{x}) = \delta^*\mathbf{u}^{\dagger}G\mathbf{x}=0. 
\end{equation}
Since $\delta$ is nonzero, we have $\mathbf{v}$ orthogonal to $\mathbf{x}$. Again, from the definition of the singular vectors
\begin{equation}
\delta\mathbf{u} =  G\mathbf{v} = Q_{\mathbf{b}}A\mathbf{v} = A\mathbf{v} - (\mathbf{b}^{\dagger}A\mathbf{v})\mathbf{b} = A\left(\mathbf{v} - (\mathbf{b}^{\dagger}A\mathbf{v})\mathbf{x}\right).
\end{equation}
The result above shows that $\mathbf{u}$ is in the image of $A$ and that $Q_\mathbf{b}\mathbf{u}=\mathbf{u}$, since $Q_{\mathbf{b}}A\left(\mathbf{v} -(\mathbf{b}^{\dagger}A\mathbf{v})\mathbf{x}\right) = \delta \mathbf{u}$. Finally, we can see from the definition of the singular values 
\begin{equation}
\delta^2 = \mathbf{u}^{\dagger}GG^{\dagger}\mathbf{u} =  \mathbf{u}^{\dagger}Q_{\mathbf{b}}AA^{\dagger}Q_{\mathbf{b}}\mathbf{u} = \mathbf{u}^{\dagger}AA^{\dagger}\mathbf{u} \leq 1/\kappa^2, 
\end{equation}
where the inequality above results from the fact that eigenvalues of $A$ lie in the interval $[\kappa^{-1},1]$.

\subsection{Kernel Projection}
The KP takes two parameters $(\eta_{\rm KP}, \kappa)$ and only works under the promise that the nonzero singular values of $G$ lie in the interval $[\kappa^{-1},1]$ which we already showed that is the case. Let us consider an arbitrary normalized $s$-qubit state $\ket{\phi}=\gamma\ket{\mathbf{w}} + \nu\ket{\mathbf{w}_\perp}$, where $\mathbf{w}$ is a unit vector in the kernel of $G$, which typically for us it will be $\ket{\mathbf{w}}=\ket{\mathbf{x}}$, and $\mathbf{w}_{\perp}$ a unit vector orthogonal to the kernel of $G$. The action of KP is to decrease the components of the state $\ket{\phi}$ which do not belong to the kernel of $G$, i.e., under its action, we have
\begin{equation}
\label{eq:KP_ex1}
\gamma\ket{\mathbf{w}} + \nu\ket{\mathbf{w}_\perp} \rightarrow \gamma\ket{\mathbf{w}} + \nu\delta_1\ket{\mathbf{w}_\perp} + \nu\delta_2\ket{\mathbf{w'}_\perp},
\end{equation}
where $\delta_1$ and $\delta_2$ are real parameters  that satisfy $\sqrt{\delta_1^2+\delta_2^2}\leq \eta_{\rm KP}$ and $\mathbf{w}'_{\perp}$ another unit vector orthogonal to the kernel of $G$. We have the probability that KP succeeds by its squared-norm $|\gamma|^2 + |\nu|^2(\delta_1^2+\delta_2^2)$.

The KP here is the same as the eigenstate filtering from  \cite{lin2020optimal}, which is then constructed through the QSVT. The basic idea is that we can build a polynomial filter function that projects onto the kernel of the matrix. From \cite{lin2020optimal}, the filter function, considering the target parameters $\widetilde{\Delta}\in(0,1)$ and $\eta_{\rm KP}\in(0,1]$   is given by the following degree-$2l$ polynomial
\begin{equation}
F_{\widetilde{\Delta},l} = \frac{T_l\left(\frac{1+\widetilde{\Delta}^2 -2x^2}{1-\widetilde{\Delta}^2}\right)}{T_l\left(\frac{1+\widetilde{\Delta}^2}{1-\widetilde{\Delta}^2}\right)},   
\end{equation}
where $T_l$ is the $l$th Chebyshev polynomial of the first kind, where its order depends on the target parameters, i.e.,
\begin{equation}
 l \leq \biggl \lceil \frac{1}{2\widetilde{\Delta}} \ln\left(\frac{2}{\eta_{\rm KP}}\right)\biggl\rceil.  
\end{equation}
We have from the Lemma 1 of \cite{dalzell2024shortcut} which is a Lemma improved from \cite{lin2020optimal} that the polynomial $F_{\widetilde{\Delta},l}$ is guaranteed to satisfy the following three properties
\begin{enumerate}
    \item For all $x \in [-1,1]$, it holds that $|F_{\widetilde{\Delta},l}(x)|\leq 1$.
    \item For all $x \in [\widetilde{\Delta},1]$, it holds that $|F_{\widetilde{\Delta},l}(x)|\leq T_l\left(\frac{1+\widetilde{\Delta}^2}{1-\widetilde{\Delta}^2}\right)^{-1}\leq \eta_{\rm KP}$.
    \item $F_{\widetilde{\Delta},l}(0)=1$.
\end{enumerate}
The properties outlined above indicate that $F_{\widetilde{\Delta},l}(x)$ functions as a filter with width $\widetilde{\Delta}$, mapping the zero input to one and assigning values close to zero for all inputs outside the interval $[-\widetilde{\Delta}, \widetilde{\Delta}]$. The idea is to apply the QSVT using the polynomial $F_{\widetilde{\Delta},t}$.

Following the example provided by \cite{dalzell2024shortcut}, we consider a more detailed explanation of what is presented in \cref{eq:KP_ex1}, including the QSVT framework. Consider the state to which we apply the KP is given by $\ket{\phi} = c_1 \ket{\mathbf{a}} + c_2 \ket{\mathbf{b}}$ where $\mathbf{a}$ is a vector in the kernel of $H$, and $\mathbf{b}$ is orthogonal to the kernel. We also assume that the vector $\ket{\mathbf{b}}$ has the following decomposition $\sum_jw_j\ket{\mathbf{u}_j}$ into right singular vectors $\ket{\mathbf{u}_j}$ of $H$ all of which have singular values $\gamma_j \in (\widetilde{\Delta},1)$. Thus, the second property of the filter, we know that of $F_{\widetilde{\Delta},l}(\gamma_j)$ lies in the interval $[-\eta_{\rm KP},\eta_{\rm KP}]$. 

In order to look at the action of the polynomial  $F_{\widetilde{\Delta},l}$ of the matrix $H$, which has the following singular value transformation (SVD) $H=W\Sigma V^{\dagger}$ where $\delta$ is the diagonal matrix of singular values, let us look the overall picture of its QSVT construction. The key result is obtained from the measurement on the state $\ket{0}$ of the ancillas, expressed as:
\begin{align}
&VF_{\widetilde{\Delta},l}(H)V^{\dagger} \nn
&= \left(\bra{0}^{\otimes(a_H+1)}\otimes I_{2^s}\right)U_{F_{\widetilde{\Delta},l}}^{(H)} \left(\ket{0}^{\otimes(a_H+1)}\otimes I_{2^s}\right),      
\end{align}
where the unitary $U_{F_{\widetilde{\Delta},l}}^{(H)}$ represents the unitary produced by the QSVT of $U_H$. The operator $U_H$ is an $(1,a_H)$-block-encoding of $H$ and its implementation involves:

\begin{itemize}
    \item $\frac{d}{2}$ calls to each of $U_H$ and $U_H^{\dagger}$,
    \item $2d$ multi-controlled Toffoli gates,
    \item $d$ single-qubit rotation gates.
\end{itemize}
Thus, the application of $VF_{\widetilde{\Delta},l}(H)V^{\dagger}$ maps
\begin{equation}
\label{eq:KP_b}
\ket{\mathbf{b}} \rightarrow \sum_{j}w_jF_{\widetilde{\Delta},l}(\gamma_j)\ket{\mathbf{u}_j} = \delta_1\ket{\mathbf{b}} + \delta_2\ket{\mathbf{b}'},    
\end{equation}
where $\mathbf{b}'$ is orthogonal to $\mathbf{b}$ and to the kernel of $H$, and $\delta^2_1 + \delta^2_2\leq \eta_{\rm KP}^2$. This allows us to assert that
\begin{align}
\label{eq:Ukproj}
U_{F_{\widetilde{\Delta},l}}^{(H)}\left(\ket{0}^{\otimes(a_H+1)}\otimes\ket{\phi}\right) &= \biggl[\ket{0}^{\otimes(a_H+1)}\otimes\left(c_1\ket{\mathbf{a}}\right. \nonumber\\
&\quad \left.+ c_2\delta_1\ket{\mathbf{b}} +c_2\delta_2 \ket{\mathbf{b}'}\right) \biggl] \nonumber\\
&\quad+ c_2\sqrt{1-\delta_1^2 -\delta_2^2}\ket{\perp},    \end{align}
where $\ket{\perp}$ is a normalized state such that $\left(\bra{0}^{\otimes(a_H+1)}\otimes I_{2^s}\right)\ket{\perp} = 0$. Thus, post-selecting on the first register being $\ket{0}^{\otimes(a_H+1)}$, this implements an approximate projector onto the kernel of $H$.

We will finish this section by restating Lemma 2 from \cite{dalzell2024shortcut} following its proof.

\begin{lemma}[Using KP to refine QLSS solution] Suppose $\mathbf{b}$ is in the column space of $A$, and let $\mathbf{x}$ denote the solution of minimum $\|\mathbf{x}\|$ to the equation $Ax = b$. Suppose $A$ has no singular values in the interval $\left(0,\kappa^{-1}\right)$ and let $\tilde{\rho}$ be a mixed quantum state for which $\bra{\mathbf{x}}\tilde{\rho}\ket{\mathbf{x}} = 1-\mu^2$ and $\bra{\mathbf{y}}\tilde{\rho}\ket{\mathbf{y}}$ that are in the kernel of $A$. Suppose KP is applied to approximately project $\tilde{\rho}$ onto the kernel of $G =Q_bA$ using parameters $(\kappa, \eta_{\rm KP})$. That is, we choose $l$ where 
\begin{equation}
\label{eq:order:pol}
 l \leq \biggl \lceil \frac{1}{2\widetilde{\Delta}} \ln\left(\frac{2}{\eta_{\rm KP}}\right)\biggl\rceil,     
\end{equation}
with $\widetilde{\Delta}= \kappa^{-1}$ and apply the unitary $U_{F_{\widetilde{\Delta},l}}^{(G)}$, and then measure the ancillas to determine success. KP succeeds with probability at least $1-\mu^2$, and when it succeeds it outputs a state $\tilde{\rho}_\mathrm{out}$  for which
\begin{align}
 \bra{\mathbf{x}}\tilde{\rho}_\mathrm{out}\ket{\mathbf{x}}&\geq 1 -\frac{\mu^2\eta_{\rm KP}^2}{1-\mu^2 +\mu^2\eta_{\rm KP}^2} \geq 1 - \frac{\mu^2\eta_{\rm KP}^2}{1-\mu^2}\\
 \frac{1}{2}\|\ketbra{\mathbf{x}}{\mathbf{x}} - \tilde{\rho}_\mathrm{out} \|_1 &\leq \frac{\mu \eta_{\rm KP}}{\sqrt{1-\mu^2+\mu^2\eta_{\rm KP}^2}} \leq \frac{\mu \eta_{\rm KP}}{\sqrt{1-\mu^2}}.
\end{align}
\end{lemma}
\begin{proof}
Let $\tilde{\rho} = \sum_{i}p_i\ketbra{\phi_i}{\phi_i}$ be a decomposition of $\tilde{\rho}$ as an ensemble of pure states. Then we have
\begin{equation}
\bra{\mathbf{x}}\tilde{\rho}\ket{\mathbf{x}} = \sum_ip_i|\bra{\mathbf{x}}\phi_i\rangle|^2 =1-\mu^2.   
\end{equation}
Let us consider $\ket{\mathbf{a}} = \ket{\mathbf{x}}$ in \cref{eq:Ukproj} and 
\begin{equation}
 \ket{\phi_i} = c_1\ket{\mathbf{x}} + c_2 \ket{\mathbf{b}},   
\end{equation}
and consenquently $c_1=|\bra{\mathbf{x}}\phi_i\rangle|$. Let $\ket{\psi_i}$ denote the output state when KP acts on $\ket{\phi_i}$ and we postselect on success, i.e., we have 
\begin{equation}
\ket{\psi_i} = \frac{1}{\sqrt{q_i}}\left(c_1\ket{\mathbf{x}} + c_2\delta_1\ket{\mathbf{b}} +c_2\delta_2 \ket{\mathbf{b}'}\right),
\end{equation}
and and let $q_i$ be the probability of success, i.e., $q_i= |c_1|^2 + |c_2\delta_1|^2 + |c_2\delta_2|^2$ to get the right state before flag on zero. Thus, $|\bra{\mathbf{x}}\psi_i\rangle|^2=|c_1|^2/q_i$. We then have the overall probability of success of KP given by
\begin{equation}
 \sum_j p_j q_j  \geq  \sum_j p_j |\bra{\mathbf{x}}\phi_i\rangle|^2 = 1 -\mu^2,
\end{equation}
which shows the first claim of this lemma. Now, from the expression of $q_i$ we know that $|\delta_1|^2 + |\delta_2|^2 \leq \eta_{\rm KP}^2$, so
\begin{align}
q_i&= |c_1|^2 + |c_2\delta_1|^2 + |c_2\delta_2|^2 \nonumber\\
&= |\bra{\mathbf{x}}\phi_i\rangle|^2 +|c_2|^2\left(|\delta_1|^2 + |\delta_2|^2\right)\nonumber\\
&=|\bra{\mathbf{x}}\phi_i\rangle|^2 +(1-|\bra{\mathbf{x}}\phi_i\rangle|^2)\left(|\delta_1|^2 + |\delta_2|^2\right)\nonumber\\
&\leq |\bra{\mathbf{x}}\phi_i\rangle|^2 +(1-|\bra{\mathbf{x}}\phi_i\rangle|^2)\eta_{\rm KP}^2. 
\end{align}
Given that the output state $\tilde{\rho}$ is 
\begin{equation}
 \tilde{\rho}_\mathrm{out} = \frac{\sum_ip_iq_i \ketbra{\psi_i}{\psi_i}   }{\sum_j p_jq_j}, 
\end{equation}
which includes the normalization. We can now combine the previous results
and see that
\begin{align}
\bra{\mathbf{x}}\tilde{\rho}_\mathrm{out}\ket{\mathbf{x}} & =   \frac{\sum_ip_iq_i |\bra{\mathbf{x}}\psi_i\rangle|^2   }{\sum_j p_jq_j}\nonumber\\
&=\frac{\sum_ip_i |\bra{\mathbf{x}}\phi_i\rangle|^2   }{\sum_j p_jq_j}\nonumber\\
&=\frac{1-\mu^2}{\sum_j p_jq_j}\nonumber\\
&\geq \frac{1-\mu^2}{\sum_j p_j\left(|\bra{\mathbf{x}}\phi_i\rangle|^2 +(1-|\bra{\mathbf{x}}\phi_i\rangle|^2)\eta_{\rm KP}^2\right)} \nonumber\\
&= 1 - \frac{\mu^2 \eta_{\rm KP}^2}{1-\mu^2+\mu^2\eta_{\rm KP}^2}.
\end{align}
The  trace distance we apply the following inequality $1/2\|\ketbra{\mathbf{x}}{\mathbf{x}} - \sigma\| \leq \sqrt{1- \bra{\mathbf{x}}\sigma\ket{\mathbf{x}}}.$
\end{proof}

\subsection{Kernel Reflection}

The KR also takes the two input parameters $(\eta_{\rm KR}, \kappa)$. The $\eta_{\rm KR}$ parameter plays the same role as $\eta_{\rm KP}$ but may be chosen to take a different value. We choose $l$ as in \cref{eq:order:pol}, with $\eta_{\rm KR}$ in place of $\eta_{\rm KP}$.  KR maps the quantum state $\ket{\phi} = \gamma\ket{\mathbf{w}} + \nu\ket{\mathbf{w}_\perp}$ as follows
\begin{equation}
\label{eq:KR_ex1}
\gamma\ket{\mathbf{w}} + \nu\ket{\mathbf{w}_\perp} \rightarrow \gamma\ket{\mathbf{w}} - \nu(1-\delta'_1)\ket{\mathbf{w}_\perp} + \nu\delta'_2\ket{\mathbf{w'}_\perp},
\end{equation}
which are the same vectors given in \cref{eq:KP_ex1}, where 
\begin{equation}
\label{eq:delta'}
\delta_1'=\frac{2(\eta_{\rm KR}+\delta_1)}{1+\eta_{\rm KR}},\quad \delta_2' = \frac{2\delta_2}{1+\eta_{\rm KR}}  ,
\end{equation}
from where the following relations can be derived
\begin{equation}
\label{eq:etas}
\delta'_1 \geq 0,\quad \sqrt{\delta'^2_1 + \delta'^2_2} \leq \frac{4\eta_{\rm KR}}{1+\eta_{\rm KR}},\;\quad |\delta_2'|\leq \frac{2\eta_{\rm KR}}{1+\eta_{\rm KR}}.    
\end{equation}

Similarly to the KP we can build a $2l$-degree polynomial within the framework of QSVT for the KR, which is given by
\begin{align}
\label{eq:K_delt}
K_{\widetilde{\Delta},l}(x) &= \frac{2F_{\widetilde{\Delta},l}(x)-1 + F_{\widetilde{\Delta},l}}{1+F_{\widetilde{\Delta},l}(\widetilde{\Delta})}\nonumber\\
&= \frac{2T_l\left(\frac{1+\widetilde{\Delta}^2 -2x^2}{1-\widetilde{\Delta}^2} + 2\right)}{T_l \left(\frac{1+\widetilde{\Delta}^2}{1-\widetilde{\Delta}^2}\right) + 1} - 1.    
\end{align}

Now from the Lemma 3 of \cite{dalzell2024shortcut} the polynomial $K_{\widetilde{\Delta},l}$ is guaranteed to satisfy the following three properties
\begin{enumerate}
    \item For all $x \in [-1,1]$, it holds that 
    \begin{equation}
        |K_{\widetilde{\Delta},l}(x)|\leq 1.
    \end{equation}
    \item For all $x \in [\widetilde{\Delta},1]$, it holds that
    \begin{equation}
     |K_{\widetilde{\Delta},l}(x)|\leq -1 + \frac{4}{T_l\left(\frac{1+\widetilde{\Delta}^2}{1-\widetilde{\Delta}^2}\right)+1}\leq  -1 + \frac{4\eta_{\rm KR}}{1+\eta_{\rm KR}}.   
    \end{equation}

    \item $K_{\widetilde{\Delta},l}(0)=1$.
\end{enumerate}

Now let us do the same exercise that we have done for KP. We consider the state $\ket{\phi} = c_1 \ket{\mathbf{a}} + c_2 \ket{\mathbf{b}}$ where $\mathbf{a}$ is a vector in the kernel of $H$, and $\mathbf{b}$ is orthogonal to the kernel. We again assume that the vector $\ket{\mathbf{b}}$ has the following decomposition $\sum_jw_j \ket{\mathbf{u}_j}$. We now look at the action of the polynomial $K_{\widetilde{\Delta},l}$ of the matrix $H$, so in this case we have
\begin{align}
&VK_{\widetilde{\Delta},l}(H)V^{\dagger} \nn
&= \left(\bra{0}^{\otimes(a_H+1)}\otimes I_{2^s}\right)U_{K_{\widetilde{\Delta},l}}^{(H)} \left(\ket{0}^{\otimes(a_H+1)}\otimes I_{2^s}\right).      
\end{align}
From \cref{eq:KP_b}, we then can see that 
\begin{align}
\ket{\mathbf{b}} &\rightarrow \sum_{j}w_jK_{\widetilde{\Delta},l}(\gamma_j)\ket{\mathbf{u}_j} \nn
&= \frac{2F_{\widetilde{\Delta},l}(x)-1 + F_{\widetilde{\Delta},l}}{1+F_{\widetilde{\Delta},l}(\widetilde{\Delta})}\ket{\mathbf{u}_j}  \nonumber\\
&= -\left(1 - \frac{2F_{\widetilde{\Delta},l}(x) }{1+F_{\widetilde{\Delta},l}(\widetilde{\Delta})}\right)\ket{\mathbf{b}} \nonumber\\
&\quad+ \frac{2}{1+F_{\widetilde{\Delta},l}(\widetilde{\Delta})}\left(\delta_1\ket{\mathbf{b}} + \delta_2\ket{\mathbf{b'}} \right)\nonumber\\
&= -\left(1 - \delta_1'\right)\ket{\mathbf{b}} + \delta_2'\ket{\mathbf{b}'}, 
\end{align}
where
\begin{equation}
    \delta'_1 = \frac{2F_{\widetilde{\Delta},l}(\widetilde{\Delta})+2\delta_1}{1+F_{\widetilde{\Delta},l}(\widetilde{\Delta})},\quad \delta_2' = \frac{2\delta_2}{1+F_{\widetilde{\Delta},l}(\widetilde{\Delta})}.
\end{equation}

We also known from the properties of $F_{\widetilde{\Delta},t}(x)$ that for all $x\in [\widetilde{\Delta},1]$, $\eta_{\rm KR} \geq |F_{\widetilde{\Delta},l}(x)|$, so if we consider $|F_{\widetilde{\Delta},l}(\widetilde{\Delta})|=\eta_{\rm KR}$ we recover \cref{eq:delta'}. Finally, we can conclude
\begin{align}
\label{eq:Ukproj2}
U_{K_{\widetilde{\Delta},l}}^{(H)}\left(\ket{0}^{\otimes(a_H+1)}\otimes\ket{\phi}\right) &= \biggl[\ket{0}^{\otimes(a_H+1)}\otimes\left(c_1\ket{\mathbf{a}}\right.  \nonumber\\ 
&\left. \quad - c_2(1 - \delta_1')\ket{\mathbf{b}} + c_2\delta'_2 \ket{\mathbf{b}'}\right) \biggl]\nonumber\\ 
&\quad+ c_2\sqrt{2\delta'_1-\delta_1^{'2} -\delta_2^{'2}}\ket{\perp}.    
\end{align}

\subsection{Comparison to QW filtering and double cost}
\label{app:QWfilt}
There are a number of subtleties in comparing the cost of the QSVT filter for the KP in the Shortcut method versus the filter used for the QW in Ref.~\cite{costa2023discrete}.
For the KP of the Shortcut method, there is a factor of $1/2$ in Eq.~\eqref{eq:order:pol}, which is not present in Eq.~(113) of Ref.~\cite{costa2023discrete}.
However, the degree of the polynomial is $2l$ for the QSVT filter for the KP, which means that the number of calls to the block encoding is equivalent between the two methods.
Another subtlety is that the polynomial has order from $-\ell$ to $+\ell$ in Ref.~\cite{costa2023discrete}, but that does not increase the cost because one may control between the qubitised walk step and its inverse simply by controlling the reflection.

A further subtlety is that the block encoding of $G$ for the QSVT only uses a single call to the block encoding of $A$, whereas the block encoding of the Hamiltonian given in Fig.~7 of Ref.~\cite{costa2023discrete} uses one call to a controlled $A$ and another to a controlled $A^\dagger$; see Fig~\ref{fig:block}.
Those two controlled block encodings may be formulated as a single control between the block encodings of $A$ and $A^\dagger$.
In block encodings it is usually the same cost to control between a block encoding of $A$ and $A^\dagger$ as it is to just control the block encoding.
For example, if a phase factor is being applied, then it would correspond to controlling between an addition versus a subtraction into a phase gradient register, which does not increase the non-Clifford cost.

In some cases it may be appropriate to cost the control between $A$ and $A^\dagger$ as double the cost of simply controlling $A$, as also reported in this section.
In that case, the cost of the QW would be doubled, but the cost of the filter would not.
The reason is that the filter is for the final Hamiltonian, in which case no superposition is needed for the qubit $a_1$ that selects between the initial and final Hamiltonians.
As a result, the qubit $a_{h_1}$ that selects between $A$ and $A^\dagger$ is flipped every time this Hamiltonian is applied (there is an $X$ on this qubit with no control).
The quantum walk starts in a chosen computational basis state for this qubit, then each step of the walk deterministically flips this qubit.
This means that, although the block encoding has been shown with both a controlled $U_A$ and $U_A^\dagger$, in the implementation only one will be needed in each step.
Similarly, only one of the $U_b$, $U_b^\dagger$ pairs will be needed each time the Hamiltonian is block encoded.
As a result, the cost of the block encoding of the Hamiltonian is one call to the block encoding of $A$ or $A^\dagger$, as well as one $U_b$, $U_b^\dagger$ pair, and the cost is equivalent to that of the block encoding of $G$ for the QSVT.

There is one difference in the costing, which is that Ref.~\cite{costa2023discrete} uses an LCU approach to the filter enabling any failure cases to be detected early.
In comparison, for the QSVT the entire QSVT needs to be performed before failure can be detected.
This means that, in the case of failure, the LCU approach has half the cost on average.
To fairly compare the two approaches, we have assumed this saving when testing both the Shortcut and QW methods.
A further subtlety is that, although the filter can be written as $\epsilon T_\ell(\beta\cos(\phi))$ where the Chebyshev polynomial has exponentially large coefficients, when this polynomial is expanded into the basis of $\cos(m\phi)$, the coefficients are positive and sum to 1.
As a result the LCU approach is efficient, and the filter succeeds with probability 1 for the ground state.

We then finish this section reporting all the cases tested previously for the QW method (except the sparse-matrix results), namely for PD and Non-Hermitian matrices for $32\times 32$ and $64\times 64$, when we count as twice the number of application of $U_A$ in the adiabatic part for the gates of the form $\ket{0}\bra{0} \otimes U_A + \ket{1}\bra{1}\otimes U_{A^\dagger}$.

\begin{figure}
\includegraphics[width = 0.5\textwidth]{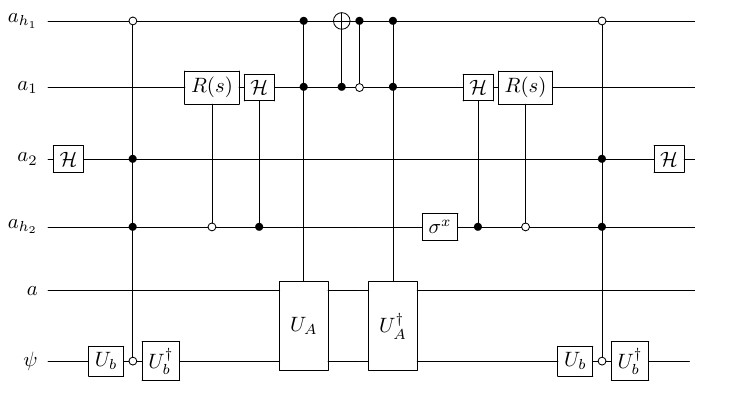}
\caption{\label{fig:block} The block encoding of the Hamiltonian from Ref.~\cite{costa2022optimal}.}
\end{figure}

\begin{figure}[H]
    \centering
\includegraphics[width=0.45\textwidth]{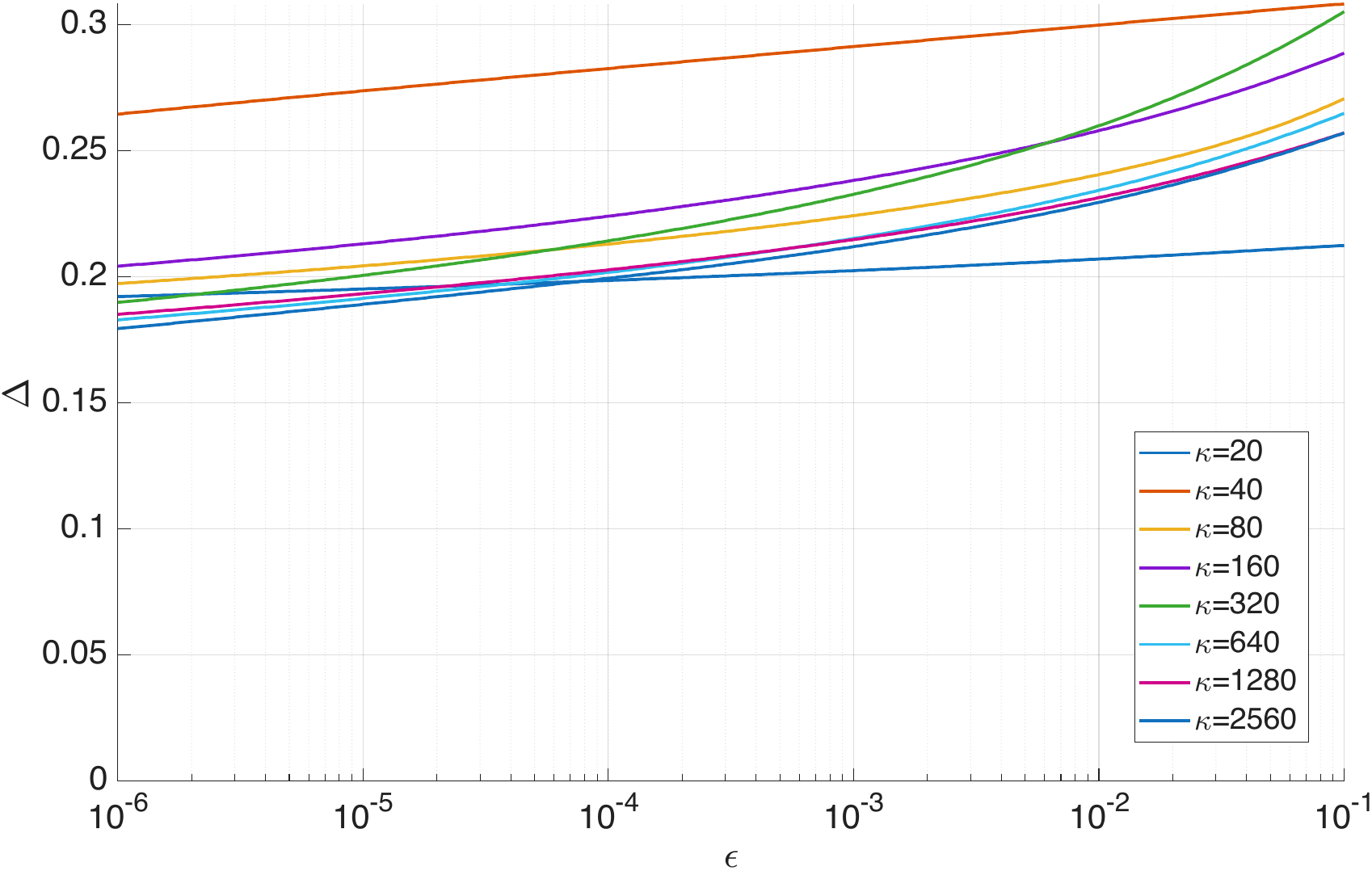}
    \caption{The plot shows the recommended values of $\Delta$, for the QW method, when we double the cost for controlling $A$, as a function of the tested condition number, for the set of positive definite matrices $A\in\mathbb{C}^{32\times 32}$.}
    \label{fig:Delta32double}
\end{figure}

\begin{figure}[H]
    \centering
\includegraphics[width=0.45\textwidth]{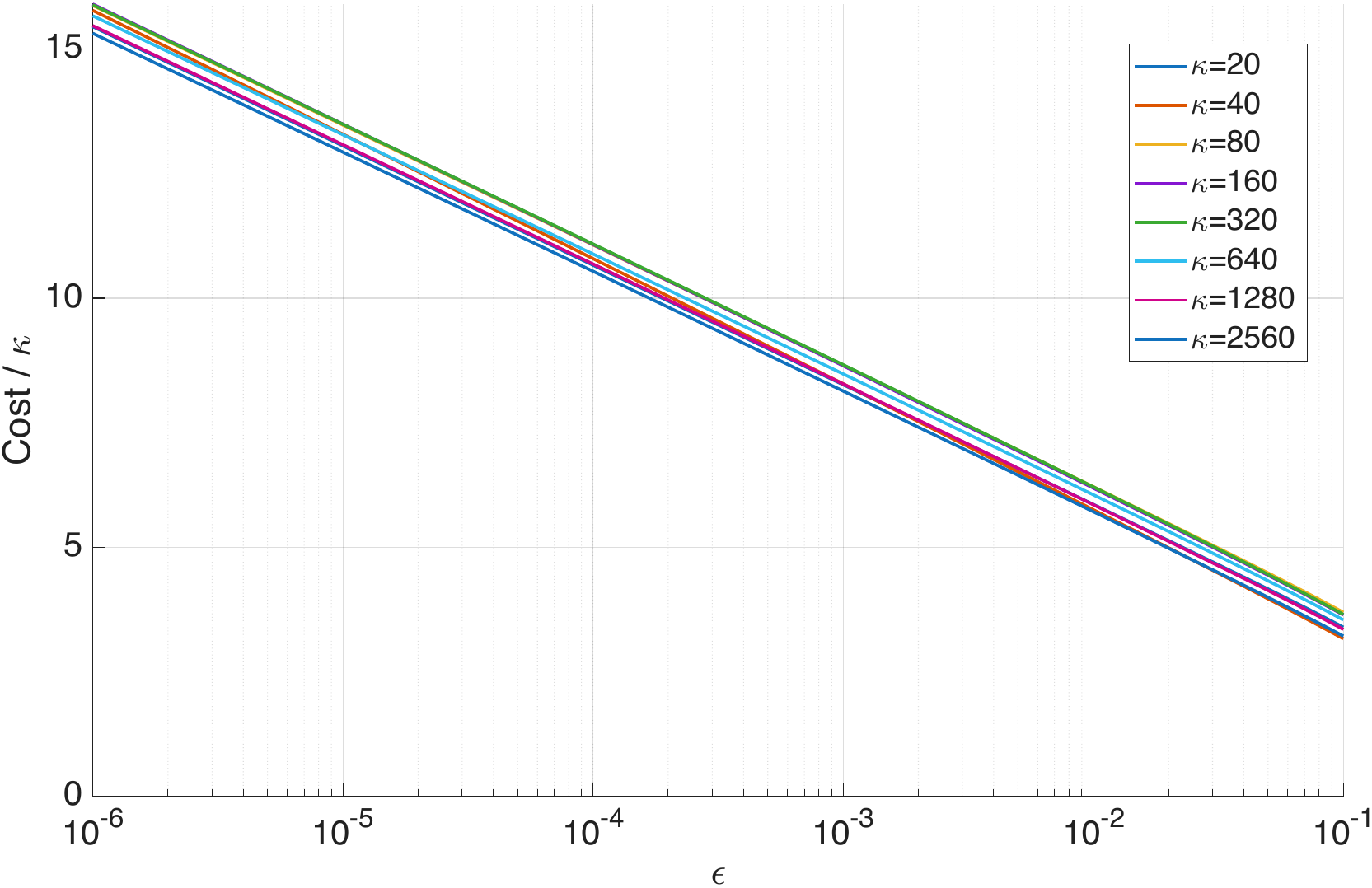}
    \caption{Total cost of the QW method for different target solution-error tolerances,  when we double the cost for controlling $A$ (in contrast for what is given in \cref{fig:Cost32}), using the optimal step sizes $\Delta$ (as estimated from \cref{fig:Delta32double}) for each tested condition number, over the set of positive definite matrices $A\in\mathbb{C}^{32\times 32}$.}
    \label{fig:Cost32double}
\end{figure}

\begin{figure}[H]
    \centering
\includegraphics[width=0.45\textwidth]{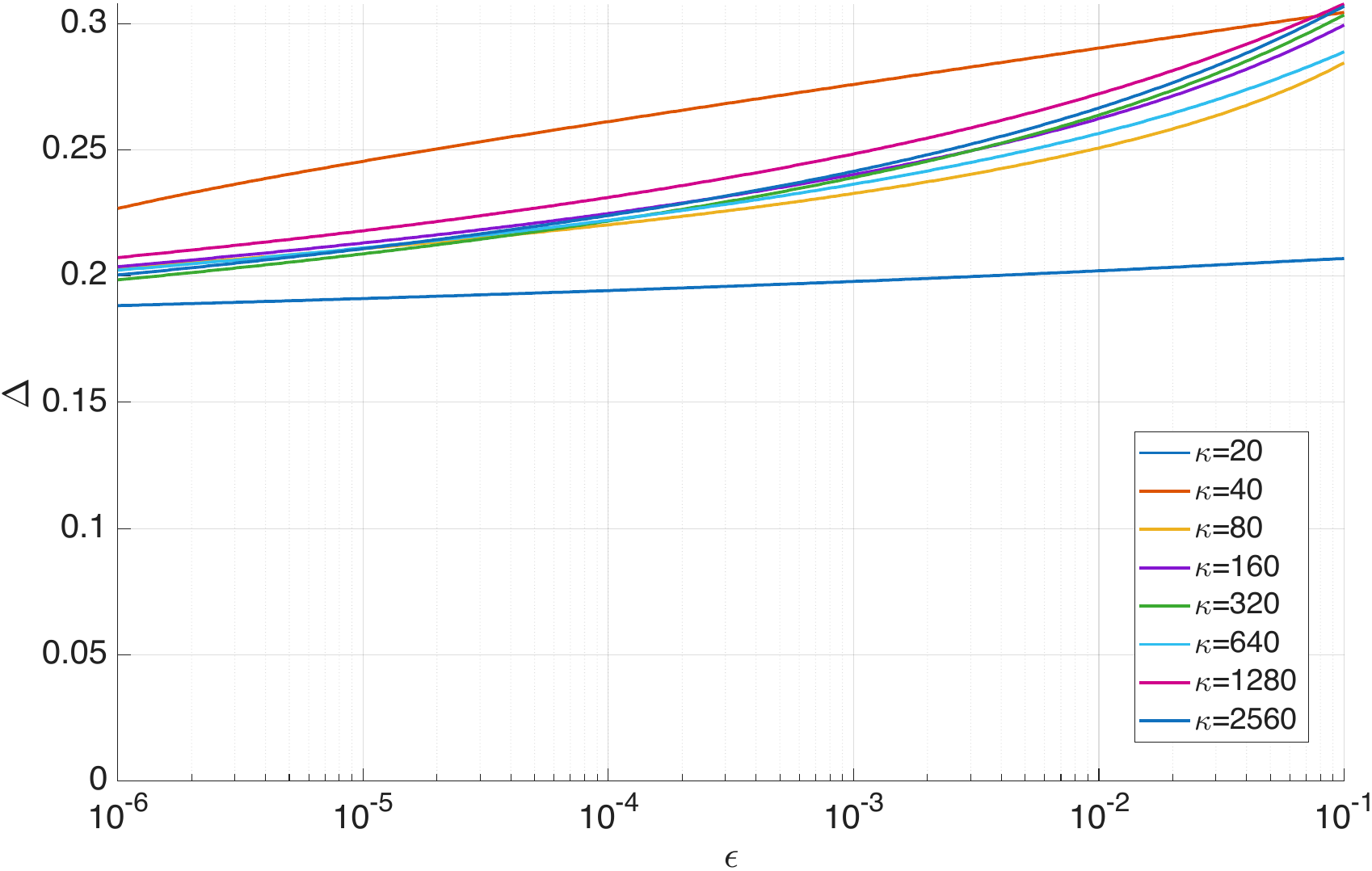}
    \caption{The plot shows the recommended values of $\Delta$ for the QW method, when we double the cost for controlling $A$, as a function of the tested condition number, for the set of positive definite matrices $A\in\mathbb{C}^{64\times 64}$.}
    \label{fig:Delta64double}
\end{figure}

\begin{figure}[H]
    \centering
\includegraphics[width=0.45\textwidth]{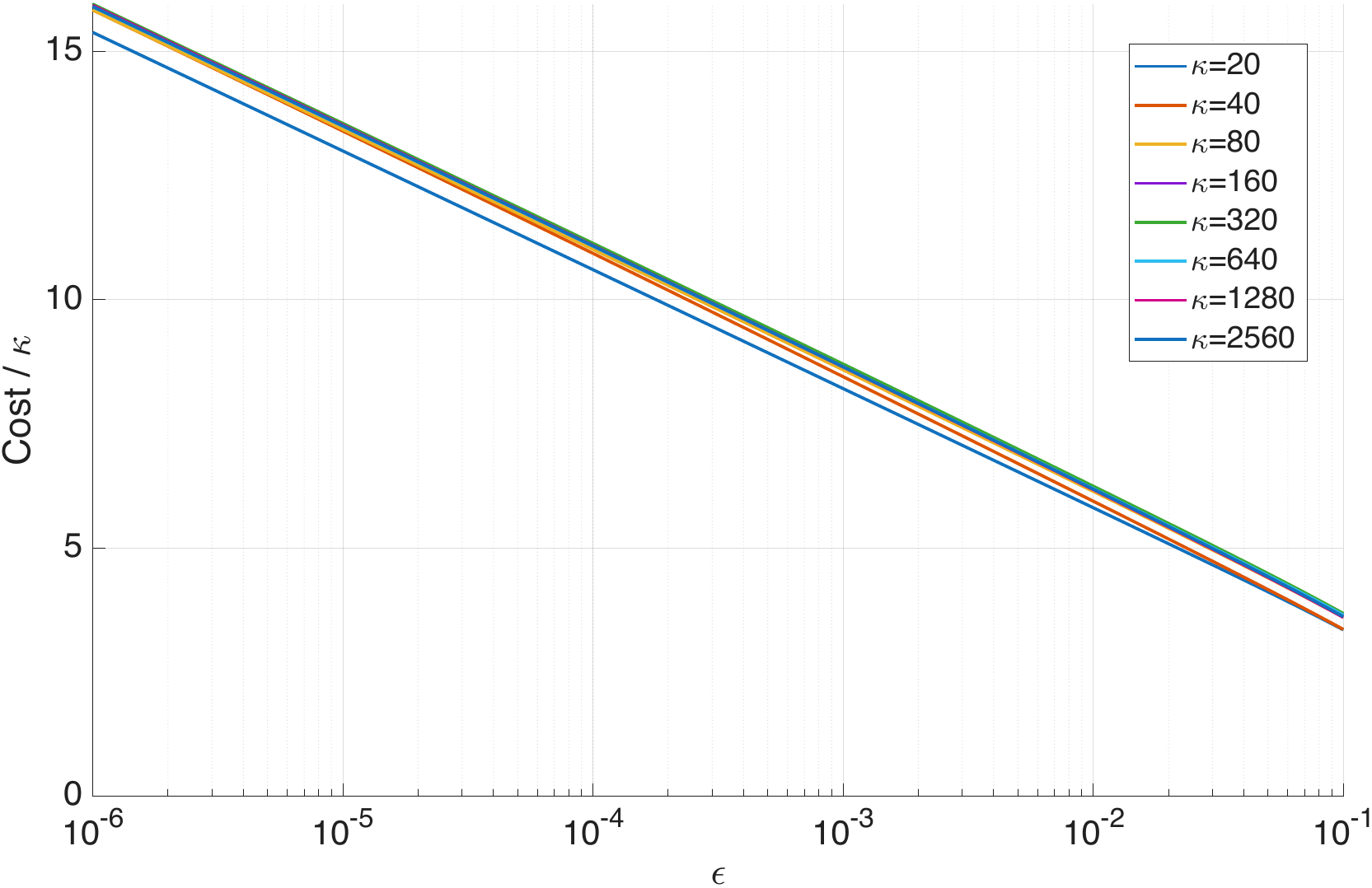}
    \caption{Total cost of the QW method for different target solution-error tolerances, when we double the cost for controlling $A$ (in contrast for what is given in \cref{fig:Cost64}), using the optimal step sizes $\Delta$ (as estimated from \cref{fig:Delta64double}) for each tested condition number, over the set of positive definite matrices $A\in\mathbb{C}^{64\times 64}$.}
    \label{fig:Cost64double}
\end{figure}

\begin{figure}[H]
    \centering
\includegraphics[width=0.45\textwidth]{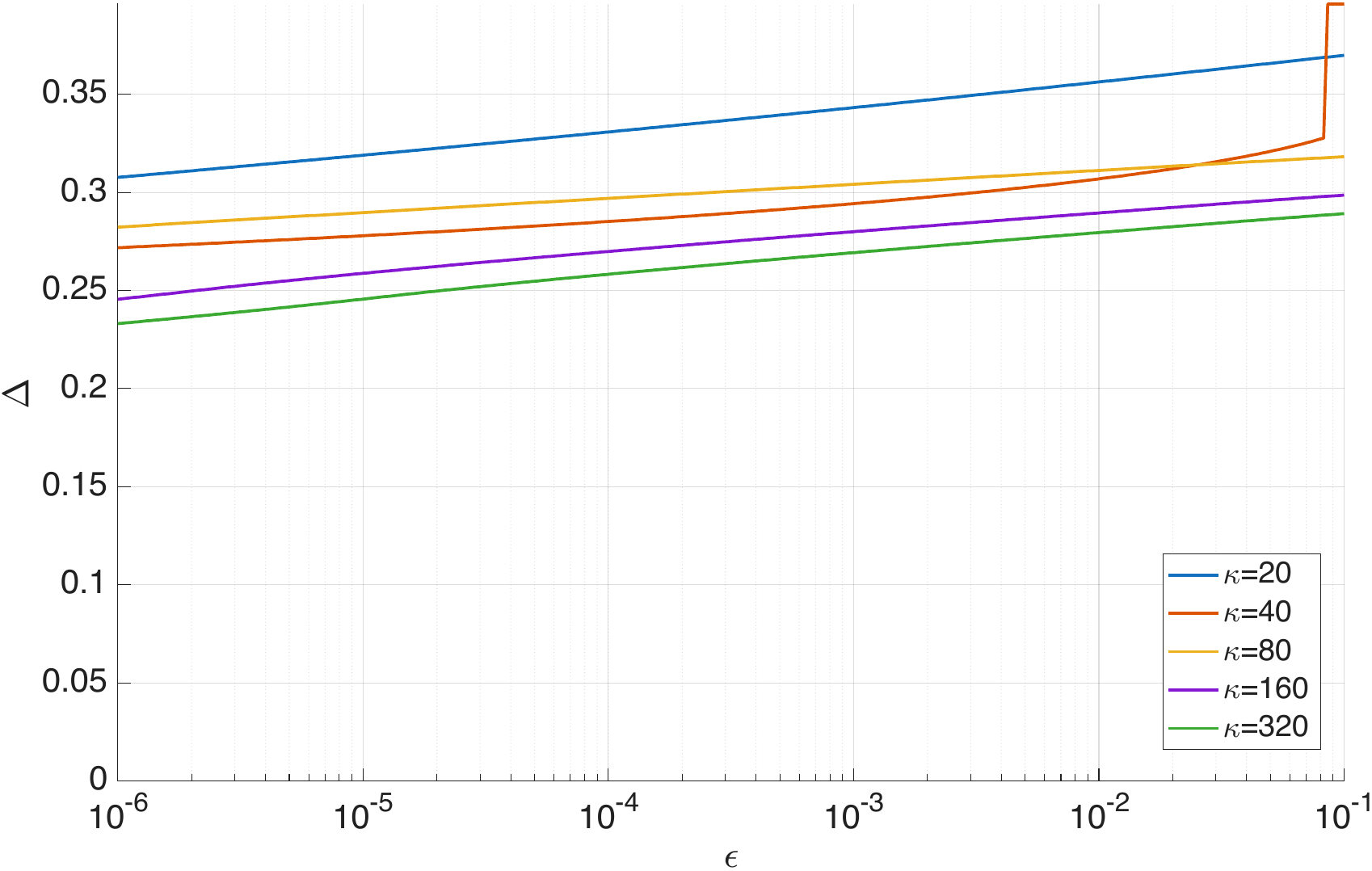}
    \caption{The plot shows the recommended values of $\Delta$ for the QW method, when we double the cost for controlling $A$, as a function of the tested condition number, for the set of non-Hermitian matrices  $A\in\mathbb{C}^{32\times 32}$.}
    \label{fig:delta_NH32_QWdouble}
\end{figure}

\begin{figure}[H]
    \centering
\includegraphics[width=0.45\textwidth]{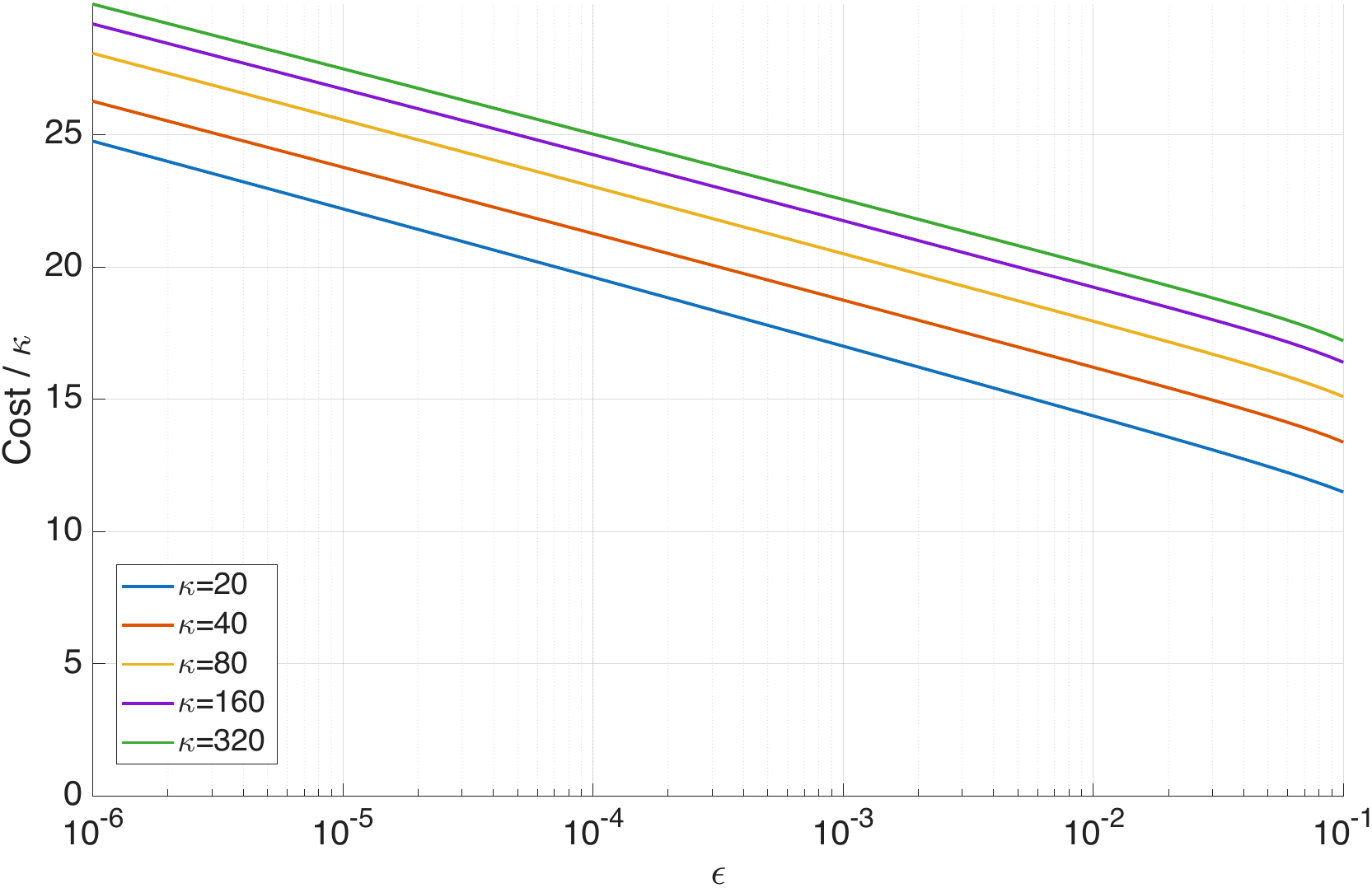}
    \caption{Total cost for the QW method for different target solution-error tolerances, when we double the cost for controlling $A$ (in contrast for what is given in \cref{fig:Cost_QW_NH32}), using the optimal step sizes $\Delta$ (as estimated from \cref{fig:delta_NH32_QWdouble}) for each tested condition number, over the set of non-Hermitian matrices $A\in\mathbb{C}^{32\times 32}$.}
    \label{fig:Cost_QW_NH32double}
\end{figure}

\begin{figure}[H]
    \centering
\includegraphics[width=0.45\textwidth]{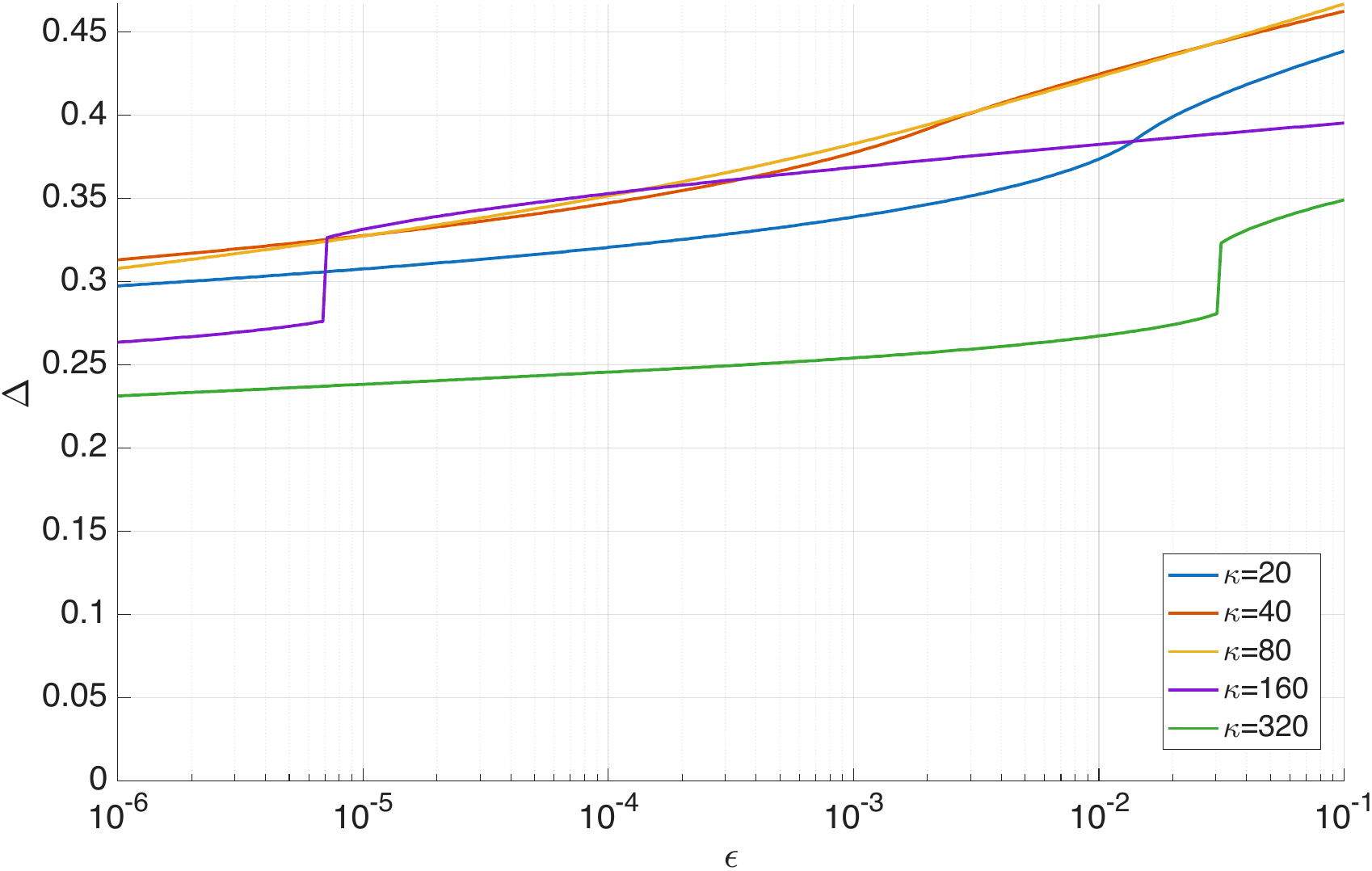}
    \caption{The plot shows the recommended values of $\Delta$ for the QW method, when we double the cost for controlling $A$, as a function of the tested condition number, for the set of non-Hermitian matrices  $A\in\mathbb{C}^{64\times 64}$.}
    \label{fig:delta_NH64_QWdouble}
\end{figure}

\begin{figure}[H]
    \centering
\includegraphics[width=0.45\textwidth]{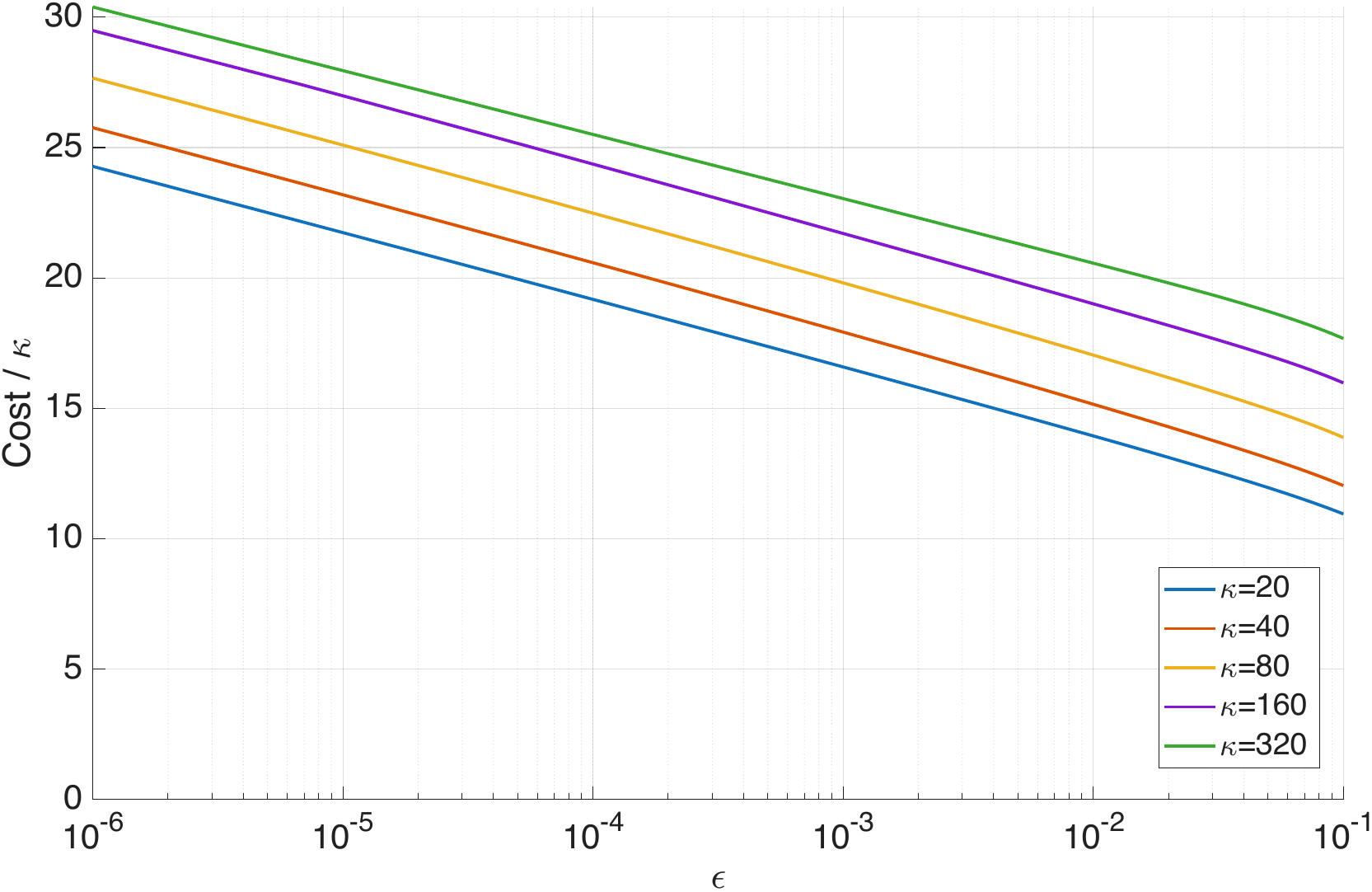}
    \caption{Total cost for the QW method for different target solution-error tolerances, when we double the cost for controlling $A$ (in contrast for what is given in \cref{fig:Cost_QW_NH64}), using the optimal step sizes $\Delta$ (as estimated from \cref{fig:delta_NH64_QWdouble} for each tested condition number, over the set of non-Hermitian matrices $A\in\mathbb{C}^{64\times 64}$.}
    \label{fig:Cost_QW_NH64double}
\end{figure}

\section{Numerical Tests - Supporting Material}
\label{app:support}

In this section, we extend further results considered in \cite{costa2023discrete} for the QW and the Randomised method. We tested both approaches, considering different dimensions across different condition numbers. Moreover, we give the plots for the recommended values for $\Delta$ accordingly to a given target precision for $\epsilon$ for the solution error for both methods, i.e., QW and Shortcut for the non-Hermitian matrices.

\subsection{Recommended values for $\Delta$}\label{app:recom_delta}

We begin by presenting the results for the QW method applied to non-Hermitian matrices, including the tables used to interpolate the cost values as a function of $\Delta$. We also provide the corresponding tables for the positive-definite (PD) cases. These results then allow us to determine the recommended values of $\Delta$ that minimize the total cost, namely the combined adiabatic and filtering contributions.

\begin{figure}[H]
    \centering
\includegraphics[width=0.45\textwidth]{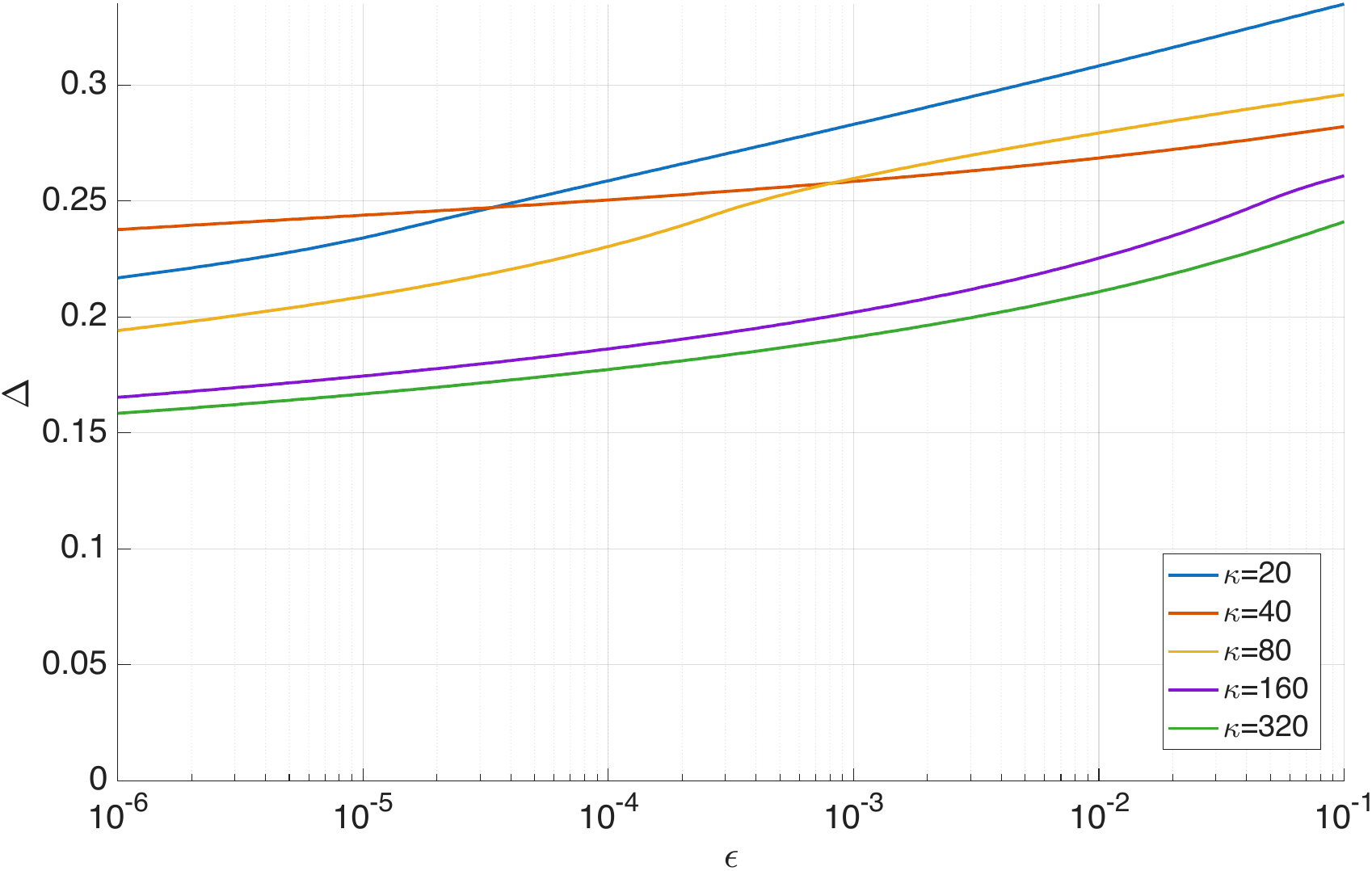}
    \caption{The plot shows the recommended values of $\Delta$ (used to estimated the total cost in \cref{fig:Cost_QW_NH32} ) resulted from \cref{tab:costNH_64}, for the QW method as a function of the total allowable error in the solution $\epsilon$ for different tested condition number, for the set of non-Hermitian matrices  $A\in\mathbb{C}^{32\times 32}$.}
    \label{fig:delta_NH32_QW}
\end{figure}

\begin{figure}[H]
    \centering
\includegraphics[width=0.45\textwidth]{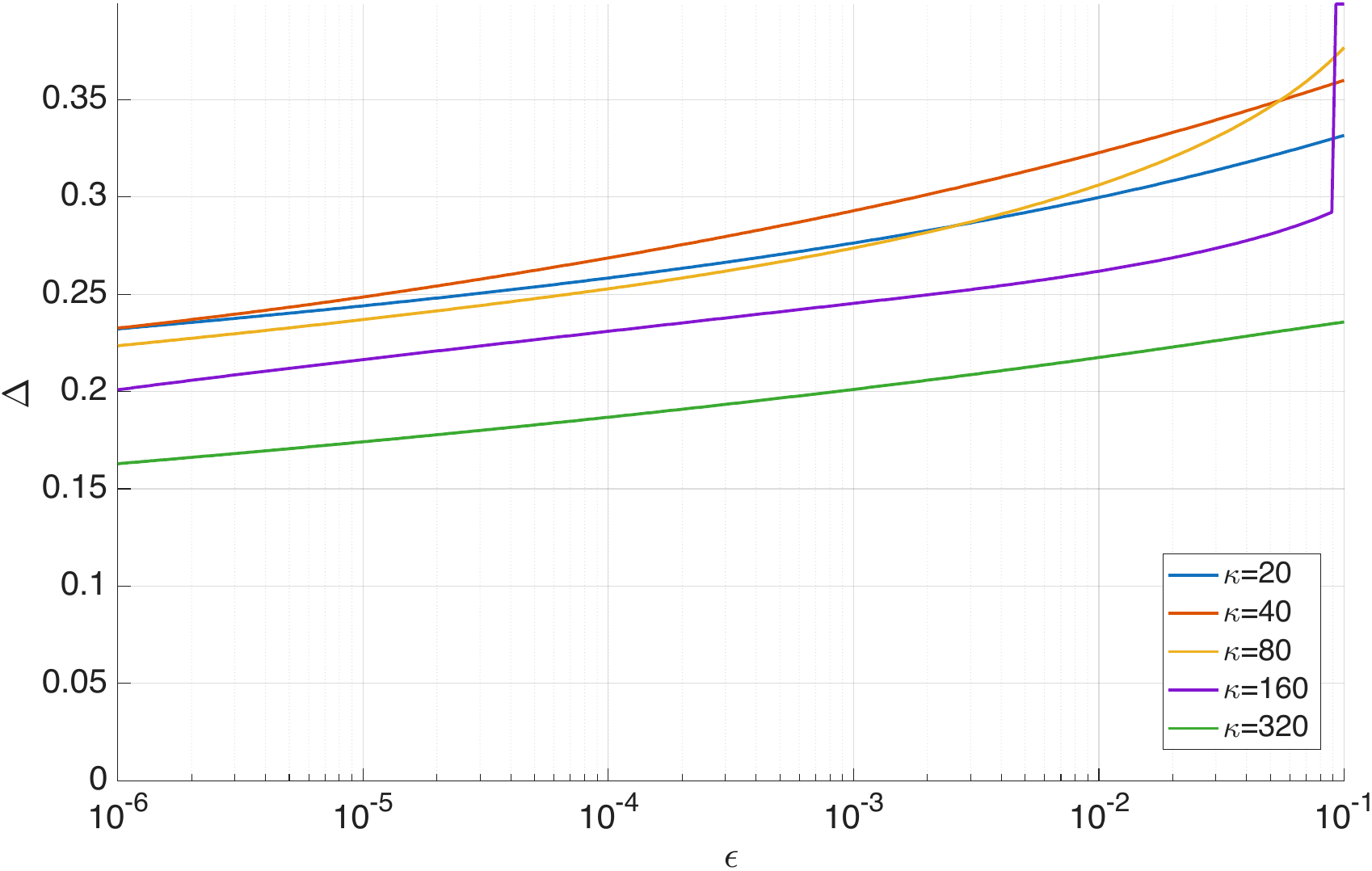}
    \caption{The plot shows the recommended values of $\Delta$ (used to estimated the total cost in \cref{fig:Cost_QW_NH64}) resulted from \cref{fig:delta_NH64_QW}, for the QW method as a function of the total allowable error in the solution $\epsilon$ for different tested condition number, for the set of non-Hermitian matrices  $A\in\mathbb{C}^{64\times 64}$.}
    \label{fig:delta_NH64_QW}
\end{figure}

We now give the results for the recommended values for $\Delta$ for the Shortcut method, followed by the tables used for interpolation of $\Delta$. 

\begin{figure}[H]
    \centering
\includegraphics[width=0.45\textwidth]{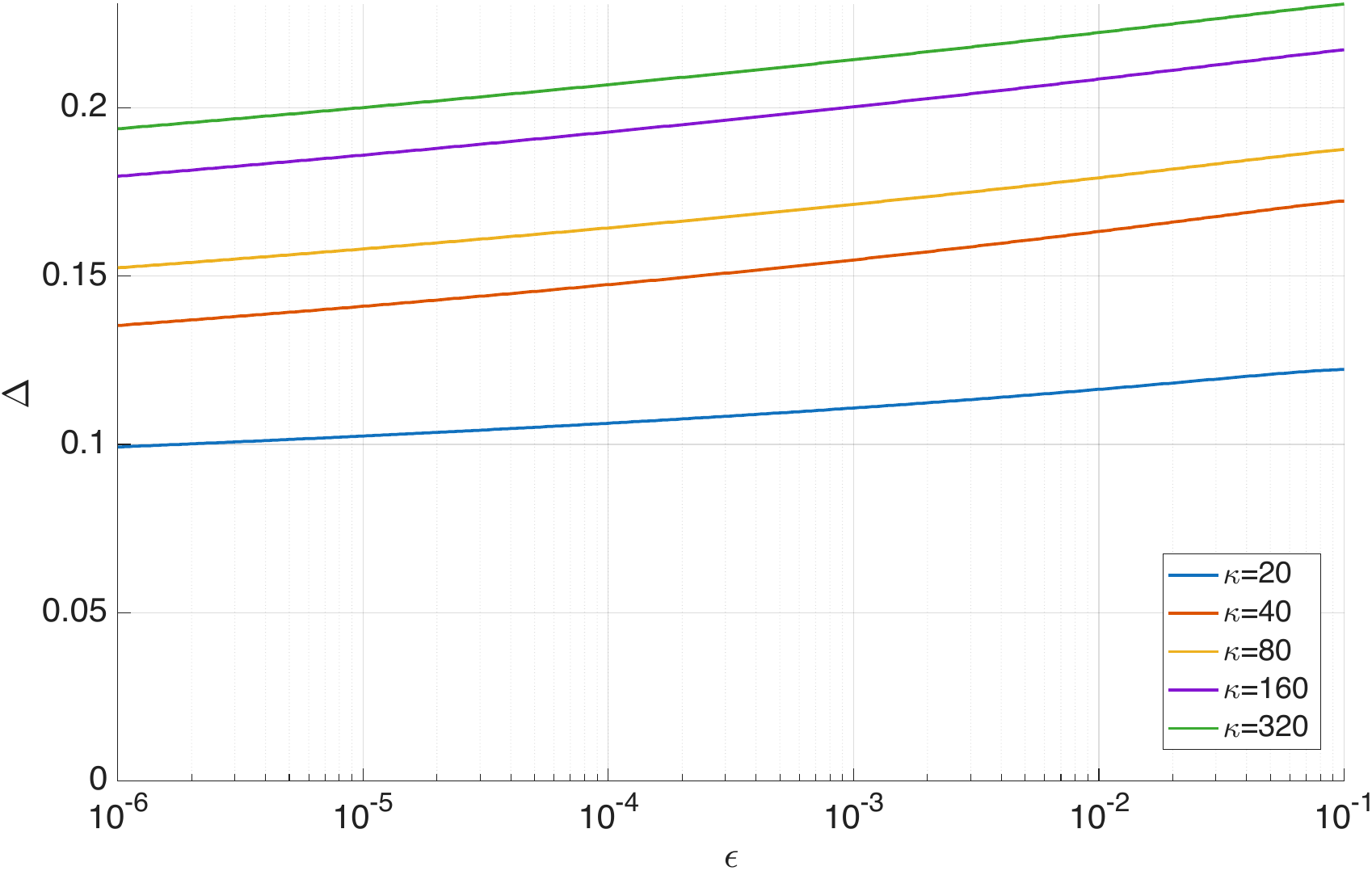}
    \caption{The plot shows the recommended values of $\Delta$ (used to estimated the total cost in \cref{fig:Cost_Sht32}) resulted from \cref{tab:cost_eta_32_unknown_norm}, for the Shortcut method as a function of the total allowable error in the solution $\epsilon$ for different tested condition number, for the set of non-Hermitian matrices  $A\in\mathbb{C}^{32\times 32}$.}
    \label{fig:recomenSH_delta32}
\end{figure}

\begin{figure}[H]
    \centering
\includegraphics[width=0.45\textwidth]{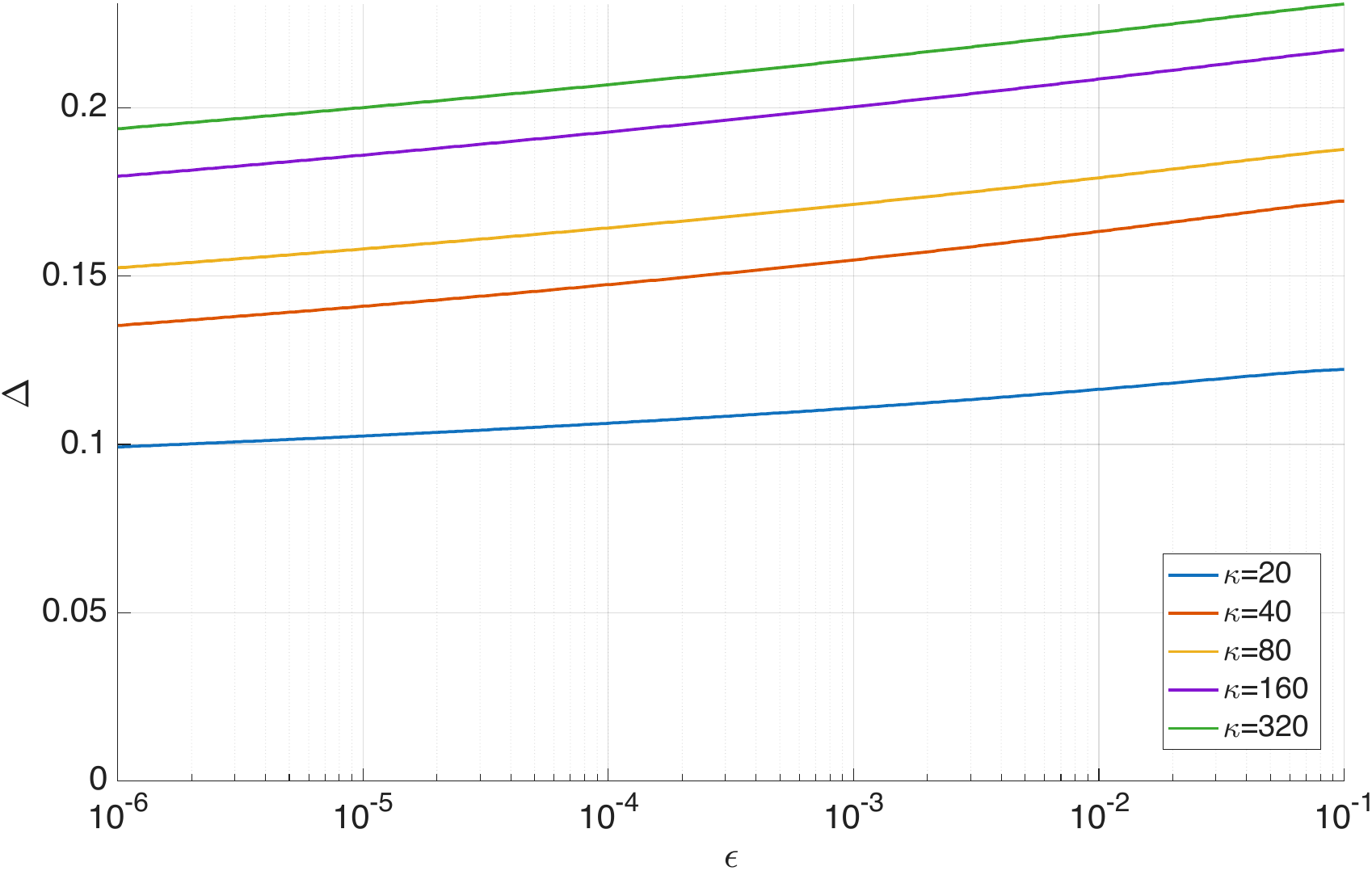}
    \caption{The plot shows the recommended values of $\Delta$ (used to estimated the total cost in \cref{fig:Cost_Sh_NH64}) resulted from \cref{tab:cost_eta_64_unknown_norm}, for the Shortcut method as a function of the total allowable error in the solution $\epsilon$ for different tested condition number, for the set of non-Hermitian matrices  $A\in\mathbb{C}^{64\times 64}$.}
    \label{fig:delta_NH64_Sh}
\end{figure}

\begin{table*}[t!]
\centering
\small
\setlength{\tabcolsep}{4pt}
\begin{tabular}{|r|c|c|c|c|c|c|c|c|c|c|}
\hline
$\kappa$ &
$\underset{\Delta=0.40}{\text{Cost}}$ & $\underset{\Delta=0.40}{\mathrm{Error}}$ &
$\underset{\Delta=0.30}{\text{Cost}}$ & $\underset{\Delta=0.30}{\mathrm{Error}}$ &
$\underset{\Delta=0.25}{\text{Cost}}$ & $\underset{\Delta=0.25}{\mathrm{Error}}$ &
$\underset{\Delta=0.20}{\text{Cost}}$ & $\underset{\Delta=0.20}{\mathrm{Error}}$ &
$\underset{\Delta=0.15}{\text{Cost}}$ & $\underset{\Delta=0.15}{\mathrm{Error}}$ \\
\hline
\hline
20  & 68   & 0.385 & 84   & 0.299 & 100  & 0.238 & 112  & 0.198 & 132  & 0.149 \\ \hline
40  & 160  & 0.396 & 196  & 0.295 & 220  & 0.244 & 252  & 0.200 & 292  & 0.148 \\ \hline
80  & 400  & 0.396 & 452  & 0.294 & 504  & 0.246 & 560  & 0.197 & 628  & 0.148 \\ \hline
160 & 888  & 0.399 & 980  & 0.299 & 1068 & 0.250 & 1172 & 0.199 & 1288 & 0.150 \\ \hline
320 & 1888 & 0.399 & 2072 & 0.299 & 2256 & 0.246 & 2448 & 0.198 & 2676 & 0.150 \\
\hline
\end{tabular}
\caption{Cost and average $\Delta$ error (QW) for different condition numbers $\kappa$ for non-Hermitian matrices $A\in\mathbb{R}^{32\times32}$  used to generate the recommended values for $\Delta$ in \cref{fig:delta_NH32_QW}.}
\label{tab:costNH_32}
\end{table*}

\begin{table*}[t!]
\centering
\small
\setlength{\tabcolsep}{4pt}
\begin{tabular}{|r|c|c|c|c|c|c|c|c|c|c|}
\hline
$\kappa$ &
$\underset{\Delta=0.40}{\text{Cost}}$ & $\underset{\Delta=0.40}{\mathrm{Error}}$ &
$\underset{\Delta=0.30}{\text{Cost}}$ & $\underset{\Delta=0.30}{\mathrm{Error}}$ &
$\underset{\Delta=0.25}{\text{Cost}}$ & $\underset{\Delta=0.25}{\mathrm{Error}}$ &
$\underset{\Delta=0.20}{\text{Cost}}$ & $\underset{\Delta=0.20}{\mathrm{Error}}$ &
$\underset{\Delta=0.15}{\text{Cost}}$ & $\underset{\Delta=0.15}{\mathrm{Error}}$ \\
\hline
\hline
20   & 64   & 0.387 & 80   & 0.299 & 92   & 0.250 & 112   & 0.189 & 132   & 0.146 \\ \hline
40   & 144  & 0.393 & 188  & 0.293 & 212  & 0.250 & 244   & 0.201 & 288   & 0.149 \\ \hline
80   & 340  & 0.396 & 436  & 0.296 & 488  & 0.249 & 552   & 0.200 & 632   & 0.150 \\ \hline
160  & 788  & 0.399 & 988  & 0.300 & 1092 & 0.249 & 1212  & 0.200 & 1348  & 0.149 \\ \hline
320  & 1780 & 0.401 & 2136 & 0.300 & 2320 & 0.247 & 2516  & 0.199 & 2760  & 0.149 \\ 
\hline
\end{tabular}
\caption{Cost and average $\Delta$ error (QW) for different condition numbers $\kappa$ for non-Hermitian matrices $A\in \mathbb{R}^{64\times64}$  used to generate the recommended values for $\Delta$ in \cref{fig:delta_NH64_QW}.}
\label{tab:costNH_64}
\end{table*}

\begin{table*}[t!]
\centering
\small
\setlength{\tabcolsep}{4pt}
\begin{tabular}{|r|c|c|c|c|c|c|c|c|c|c|}
\hline
$\kappa$ &
$\underset{\Delta=0.40}{\text{Cost}}$ & $\underset{\Delta=0.40}{\mathrm{Error}}$ &
$\underset{\Delta=0.30}{\text{Cost}}$ & $\underset{\Delta=0.30}{\mathrm{Error}}$ &
$\underset{\Delta=0.20}{\text{Cost}}$ & $\underset{\Delta=0.20}{\mathrm{Error}}$ &
$\underset{\Delta=0.15}{\text{Cost}}$ & $\underset{\Delta=0.15}{\mathrm{Error}}$ &
$\underset{\Delta=0.10}{\text{Cost}}$ & $\underset{\Delta=0.10}{\mathrm{Error}}$ \\
\hline
\hline
20   & 12   & 0.3237 & 16   & 0.2422 & 20   & 0.1868 & 28   & 0.1470 & 44   & 0.0911 \\ \hline
40   & 20   & 0.3725 & 24   & 0.2780 & 48   & 0.1939 & 60   & 0.1492 & 88   & 0.0977 \\ \hline
80   & 36   & 0.3978 & 60   & 0.2924 & 92   & 0.1945 & 124  & 0.1457 & 172  & 0.1001 \\ \hline
160  & 76   & 0.4003 & 112  & 0.3009 & 180  & 0.2007 & 248  & 0.1477 & 348  & 0.0987 \\ \hline
320  & 148  & 0.3945 & 232  & 0.2940 & 360  & 0.1998 & 472  & 0.1489 & 656  & 0.0993 \\ \hline
640  & 292  & 0.3984 & 432  & 0.2970 & 680  & 0.1934 & 868  & 0.1494 & 1224 & 0.1000 \\ \hline
1280 & 528  & 0.3912 & 780  & 0.2929 & 1200 & 0.1996 & 1660 & 0.1459 & 2380 & 0.0975 \\ \hline
2560 & 896  & 0.3973 & 1368 & 0.2997 & 2240 & 0.1986 & 3000 & 0.1496 & 4080 & 0.0990 \\
\hline
\end{tabular}
\caption{Cost and average $\Delta$ error (QW) for different condition numbers $\kappa$ for positive definite matrices $A\in \mathbb{R}^{32\times32}$ used to generate the recommended values for $\Delta$ in \cref{fig:Delta32}.}
\label{tab:cost_PD_32}
\end{table*}

\begin{table*}[t!]
\centering
\small
\setlength{\tabcolsep}{4pt}
\begin{tabular}{|r|c|c|c|c|c|c|c|c|c|c|}
\hline
$\kappa$ &
$\underset{\Delta=0.40}{\text{Cost}}$ & $\underset{\Delta=0.40}{\mathrm{Error}}$ &
$\underset{\Delta=0.30}{\text{Cost}}$ & $\underset{\Delta=0.30}{\mathrm{Error}}$ &
$\underset{\Delta=0.20}{\text{Cost}}$ & $\underset{\Delta=0.20}{\mathrm{Error}}$ &
$\underset{\Delta=0.15}{\text{Cost}}$ & $\underset{\Delta=0.15}{\mathrm{Error}}$ &
$\underset{\Delta=0.10}{\text{Cost}}$ & $\underset{\Delta=0.10}{\mathrm{Error}}$ \\
\hline
\hline
20   & 12   & 0.318 & 16   & 0.237 & 20   & 0.189 & 28   & 0.143 & 44   & 0.087 \\ \hline
40   & 20   & 0.367 & 28   & 0.270 & 44   & 0.201 & 60   & 0.144 & 84   & 0.100 \\ \hline
80   & 36   & 0.388 & 56   & 0.296 & 92   & 0.192 & 124  & 0.147 & 176  & 0.098 \\ \hline
160  & 76   & 0.391 & 116  & 0.294 & 188  & 0.195 & 248  & 0.149 & 352  & 0.099 \\ \hline
320  & 152  & 0.397 & 236  & 0.294 & 380  & 0.194 & 492  & 0.149 & 700  & 0.100 \\ \hline
640  & 308  & 0.395 & 452  & 0.300 & 720  & 0.201 & 984  & 0.149 & 1412  & 0.100 \\ \hline
1280 & 576  & 0.399 & 880  & 0.300 & 1520 & 0.193 & 1980 & 0.150 & 2860 & 0.100 \\ \hline
2560 & 1140 & 0.399 & 1776 & 0.300 & 2900 & 0.200 & 3880 & 0.150 & 5560 & 0.101 \\
\hline
\end{tabular}
\caption{Cost and average $\Delta$ error (QW) for different condition numbers $\kappa$ for positive definite matrices $A\in \mathbb{R}^{64\times64}$  used to generate the recommended values for $\Delta$ in \cref{fig:Delta64}.}
\label{tab:cost_PD_64}
\end{table*}

\begin{table*}[t!]
\centering
\small
\setlength{\tabcolsep}{4pt}
\begin{tabular}{|r|ccc|ccc|ccc|ccc|}
\hline
$\kappa$ &
\multicolumn{3}{c|}{$\Delta = 0.30$} &
\multicolumn{3}{c|}{$\Delta = 0.20$} &
\multicolumn{3}{c|}{$\Delta = 0.10$} &
\multicolumn{3}{c|}{$\Delta = 0.05$} \\
\cline{2-13}
 & Cost & Error & $\eta_{\rm KR}$ &
   Cost & Error & $\eta_{\rm KR}$ &
   Cost & Error & $\eta_{\rm KR}$ &
   Cost & Error & $\eta_{\rm KR}$ \\
\hline\hline
20  &
$1.51\times 10^{2}$ & 0.2891 & 0.1488 &
$1.66\times 10^{2}$ & 0.1942 & 0.0981 &
$1.88\times 10^{2}$ & 0.0986 & 0.0450 &
$2.14\times 10^{2}$ & 0.0495 & 0.0235 \\ \hline

40  &
$3.72\times 10^{2}$ & 0.2991 & 0.1136 &
$4.16\times 10^{2}$ & 0.1858 & 0.0632 &
$4.80\times 10^{2}$ & 0.0924 & 0.0291 &
$5.48\times 10^{2}$ & 0.0446 & 0.0137 \\ \hline

80  &
$9.22\times 10^{2}$ & 0.2889 & 0.0713 &
$1.02\times 10^{3}$ & 0.1910 & 0.0420 &
$1.18\times 10^{3}$ & 0.0940 & 0.0188 &
$1.33\times 10^{3}$ & 0.0497 & 0.0097 \\ \hline

160 &
$2.20\times 10^{3}$ & 0.3015 & 0.0493 &
$2.48\times 10^{3}$ & 0.1969 & 0.0262 &
$2.90\times 10^{3}$ & 0.0958 & 0.0109 &
$3.28\times 10^{3}$ & 0.0460 & 0.0052 \\ \hline

320 &
$5.34\times 10^{3}$ & 0.2989 & 0.0296 &
$6.02\times 10^{3}$ & 0.1972 & 0.0152 &
$7.04\times 10^{3}$ & 0.0951 & 0.0061 &
$7.82\times 10^{3}$ & 0.0498 & 0.0031 \\
\hline
\end{tabular}
\caption{Average cost, average error, and $\eta_{\rm KR}$ for different condition numbers $\kappa$ and threshold values $\Delta$ for the Shortcut method for non-Hermitian matrices $A \in \mathbb{R}^{32\times 32}$ in the unknown-norm regime.}
\label{tab:cost_eta_32_unknown_norm}
\end{table*}

\begin{table*}[t!]
\centering
\small
\setlength{\tabcolsep}{4pt}
\begin{tabular}{|r|ccc|ccc|ccc|ccc|}
\hline
$\kappa$ &
\multicolumn{3}{c|}{$\Delta = 0.30$} &
\multicolumn{3}{c|}{$\Delta = 0.20$} &
\multicolumn{3}{c|}{$\Delta = 0.10$} &
\multicolumn{3}{c|}{$\Delta = 0.05$} \\
\cline{2-13}
 & Cost & Error & $\eta_{\rm KR}$ &
   Cost & Error & $\eta_{\rm KR}$ &
   Cost & Error & $\eta_{\rm KR}$ &
   Cost & Error & $\eta_{\rm KR}$ \\
\hline\hline
20  &
$1.46\times 10^{2}$ & 0.2979 & 0.1748 &
$1.62\times 10^{2}$ & 0.1974 & 0.1000 &
$1.89\times 10^{2}$ & 0.0903 & 0.0492 &
$2.16\times 10^{2}$ & 0.0457 & 0.0240 \\ \hline

40  &
$3.66\times 10^{2}$ & 0.2944 & 0.1220 &
$4.12\times 10^{2}$ & 0.1865 & 0.0712 &
$4.68\times 10^{2}$ & 0.0974 & 0.0362 &
$5.36\times 10^{2}$ & 0.0459 & 0.0165 \\ \hline

80  &
$9.00\times 10^{2}$ & 0.2887 & 0.0892 &
$9.92\times 10^{2}$ & 0.1954 & 0.0520 &
$1.15\times 10^{3}$ & 0.0911 & 0.0188 &
$1.31\times 10^{3}$ & 0.0466 & 0.0115 \\ \hline

160 &
$2.02\times 10^{3}$ & 0.2927 & 0.0619 &
$2.40\times 10^{3}$ & 0.1918 & 0.0334 &
$2.80\times 10^{3}$ & 0.0934 & 0.0148 &
$3.12\times 10^{3}$ & 0.0493 & 0.0076 \\ \hline

320 &
$5.14\times 10^{3}$ & 0.2960 & 0.0390 &
$5.80\times 10^{3}$ & 0.1957 & 0.0208 &
$6.74\times 10^{3}$ & 0.0927 & 0.0088 &
$7.52\times 10^{3}$ & 0.0484 & 0.0045 \\
\hline
\end{tabular}
\caption{Average cost, average error, and $\eta_{\rm KR}$ for different condition numbers $\kappa$ and threshold values $\Delta$ for the Shortcut method for non-Hermitian matrices $A \in \mathbb{R}^{64\times 64}$ in the unknown-norm regime.}
\label{tab:cost_eta_64_unknown_norm}
\end{table*}

\subsection{Extended analysis of the QW and Randomised method}\label{sec:further}

We report the results over 100 instances by fixing the walk step, in the QW method, so we get on average the following set of $\delta$ values: $\Delta \in \{0.4,0.3,0.25,0.20,0.15\}$, which is reported in \cref{tab:performance_04,tab:performance_metrics1,tab:performance_metrics2,tab:costNH_32,tab:costNH_64}, which represents the costly part of the quantum linear solvers, accordingly to \cite{costa2023discrete}, before the filtering step.

In the table, we report the total cost, which represents the number of applications of the block encoding of $U_A$, the average error (including the standard deviation), and the average constant factor (also including the standard deviation).

\begin{table*}[t!]
\centering
\renewcommand{\arraystretch}{1.2}
\begin{tabular}{|c||c|c|c||c|c|c|}
\hline
\multicolumn{7}{|c|}{\textbf{Results for } $A \in \mathbb{R}^{8\times 8}$} \\
\hline
\multicolumn{1}{|c||}{} & \multicolumn{3}{c||}{Quantum Walk} & \multicolumn{3}{c|}{Randomised method} \\
\hline
$\kappa$ & $\text{Cost}$ & $\text{(Error)}_{avg}$ & $\alpha_{avg}$ & $\text{(Cost)}_{avg}$ & $\text{(Error)}_{avg}$ & $\alpha_{avg}$ \\
\hline
20  & $9.20\times 10$    & $0.381 \pm 0.114$ & $1.75 \pm 0.53$ & $6.42 \times 10^2$ & $0.392 \pm 0.022$ & --- \\
40  & $2.04\times 10^2$  & $0.393 \pm 0.088$ & $2.00 \pm 0.45$ & $1.41 \times 10^3$ & $0.393 \pm 0.013$ & --- \\
80  & $4.40\times 10^2$  & $0.394 \pm 0.059$ & $2.16 \pm 0.48$ & $3.06 \times 10^3$ & $0.394 \pm 0.017$ & --- \\
160 & $9.20\times 10^2$  & $0.395 \pm 0.038$ & $2.26 \pm 0.34$ & $6.47 \times 10^3$ & $0.393 \pm 0.018$ & --- \\
\hline
\end{tabular}
\caption{Performance comparison between the discrete adiabatic (Quantum Walk) and the Randomised method across 100 non-Hermitian matrix instances of size $8 \times 8$ for the target error $\Delta=0.4$.}
\label{tab:performance_04}
\end{table*}

\begin{table*}[t!]
\centering
\renewcommand{\arraystretch}{1.2}
\begin{tabular}{|c||c|c|c||c|c|c|}
\hline
\multicolumn{7}{|c|}{\textbf{Results for } $A \in \mathbb{R}^{16\times 16}$} \\
\hline
\multicolumn{1}{|c||}{} & \multicolumn{3}{c||}{Quantum Walk} & \multicolumn{3}{c|}{Randomised method} \\
\hline
$\kappa$ & $\text{Cost}$ & $\text{(Error)}_{avg}$ & $\alpha_{avg}$ & $\text{(Cost)}_{avg}$ & $\text{(Error)}_{avg}$ & $\alpha_{avg}$ \\
\hline
20  & $7.60\times 10$   & $0.398 \pm 0.154$ & $1.51 \pm 0.59$  & $6.19\times 10^2$ & $0.391 \pm 0.029$ & --- \\
40  & $1.84\times 10^2$ & $0.399 \pm 0.142$ & $1.83 \pm 0.65$  & $1.36\times 10^3$ & $0.393 \pm 0.018$ & --- \\
80  & $4.24\times 10^2$ & $0.398 \pm 0.115$ & $2.10 \pm 0.61$  & $3.02\times 10^3$ & $0.394 \pm 0.018$ & --- \\
160 & $9.00\times 10^2$ & $0.398 \pm 0.082$ & $2.24 \pm 0.46$  & $6.56\times 10^3$ & $0.394 \pm 0.020$ & --- \\
320 & $1.88\times 10^3$ & $0.400 \pm 0.042$ & $2.35 \pm 0.245$ & $1.39\times 10^4$ & $0.395 \pm 0.026$ & --- \\
\hline
\end{tabular}
\caption{Performance comparison between the discrete adiabatic (Quantum Walk) and the Randomised method across 100 non-Hermitian matrix instances of size $16 \times 16$ for the target error $\Delta=0.4$.}
\label{tab:performance_metrics1}
\end{table*}

\begin{table*}[t!]
\centering
\renewcommand{\arraystretch}{1.2}
\begin{tabular}{|c||c|c|c||c|c|c|}
\hline
\multicolumn{7}{|c|}{\textbf{Results for } $A \in \mathbb{R}^{32\times 32}$} \\
\hline
\multicolumn{1}{|c||}{} & \multicolumn{3}{c||}{Quantum Walk} & \multicolumn{3}{c|}{Randomised method} \\
\hline
$\kappa$ & $\text{Cost}$ & $\text{(Error)}_{avg}$ & $\alpha_{avg}$ & $\text{(Cost)}_{avg}$ & $\text{(Error)}_{avg}$ & $\alpha_{avg}$ \\
\hline
20  & $6.80\times 10$   & $0.385 \pm 0.133$ & $1.31 \pm 0.45$  & $5.94 \times 10^2$ & $0.392 \pm 0.014$ & --- \\
40  & $1.60\times 10^2$ & $0.396 \pm 0.156$ & $1.58 \pm 0.62$  & $1.30 \times 10^3$ & $0.394 \pm 0.012$ & --- \\
80  & $4.00\times 10^2$ & $0.393 \pm 0.136$ & $1.97 \pm 0.68$  & $2.89 \times 10^3$ & $0.394 \pm 0.014$ & --- \\
160 & $8.88\times 10^2$ & $0.399 \pm 0.105$ & $2.19 \pm 0.58$  & $6.41 \times 10^3$ & $0.395 \pm 0.020$ & --- \\
320 & $1.88\times 10^3$ & $0.399 \pm 0.070$ & $2.34 \pm 0.41$  & $1.39 \times 10^4$ & $0.395 \pm 0.015$ & --- \\
640 & $3.88\times 10^3$ & $0.400 \pm 0.049$ & $2.42 \pm 0.30$  & $2.96 \times 10^4$ & $0.395 \pm 0.011$ & --- \\
\hline
\end{tabular}
\caption{Performance comparison between the discrete adiabatic (Quantum Walk) and the Randomised method across 100 non-Hermitian matrix instances of size $32 \times 32$ for the target error $\Delta=0.4$.}
\label{tab:performance_metrics2}
\end{table*}

\end{document}